\def\year{2021}\relax
\documentclass[letterpaper]{article} % DO NOT CHANGE THIS
\usepackage{aaai21}  % DO NOT CHANGE THIS
\usepackage{times}  % DO NOT CHANGE THIS
\usepackage{helvet} % DO NOT CHANGE THIS
\usepackage{courier}  % DO NOT CHANGE THIS
\usepackage[hyphens]{url}  % DO NOT CHANGE THIS
\usepackage{graphicx} % DO NOT CHANGE THIS
\usepackage{natbib}  % DO NOT CHANGE THIS AND DO NOT ADD ANY OPTIONS TO IT
\usepackage{caption} % DO NOT CHANGE THIS AND DO NOT ADD ANY OPTIONS TO IT
\usepackage{amssymb,amsmath,amsthm,amsfonts}
\usepackage{mathtools}
\usepackage{enumitem}
\usepackage{graphicx}
\usepackage[font=small]{caption}
\usepackage[labelformat=simple]{subcaption}
\usepackage{float}
\usepackage[ruled,vlined,linesnumbered]{algorithm2e}
\usepackage{braket}
\usepackage{footnote}
\usepackage{xcolor}
\usepackage{dsfont}

\usepackage{tikz}
\usetikzlibrary{arrows.meta}
\usetikzlibrary{positioning}

\usetikzlibrary{calc,patterns,angles,quotes}
\usepackage{siunitx}

\usepackage{hyperref}

\newtheorem{Lemma}{Lemma}

\newtheorem{Theorem}{Theorem}
\newtheorem{Definition}{Definition}

\newtheorem{Corollary}{Corollary}

\usepackage{thm-restate}

\newcommand{\eq}[1]{(\ref{eq:#1})}

\newcommand{\thm}[1]{\hyperref[thm:#1]{Theorem~\ref*{thm:#1}}}
\newcommand{\cor}[1]{\hyperref[cor:#1]{Corollary~\ref*{cor:#1}}}
\newcommand{\defn}[1]{\hyperref[defn:#1]{Definition~\ref*{defn:#1}}}
\newcommand{\lem}[1]{\hyperref[lem:#1]{Lemma~\ref*{lem:#1}}}
\newcommand{\prop}[1]{\hyperref[prop:#1]{Proposition~\ref*{prop:#1}}}
\newcommand{\fig}[1]{\hyperref[fig:#1]{Figure~\ref*{fig:#1}}}
\newcommand{\tab}[1]{\hyperref[tab:#1]{Table~\ref*{tab:#1}}}
\newcommand{\algo}[1]{\hyperref[algo:#1]{Algorithm~\ref*{algo:#1}}}
\newcommand{\fnote}[1]{\hyperref[fnote:#1]{Footnote~\ref*{fnote:#1}}}
\newcommand{\fac}[1]{\hyperref[fac:#1]{Fact~\ref*{fac:#1}}}
\newcommand{\lin}[1]{\hyperref[lin:#1]{Line~\ref*{lin:#1}}}

\SetKwInput{KwInput}{Input}
\SetKwInput{KwOutput}{Output}

\def\trans{^{\top}}

\newcommand{\range}[1]{[#1]}

\newcommand{\A}{\mathcal{A}}
\renewcommand{\H}{\mathcal{H}}

\newcommand{\C}{\mathbb{C}}

\newcommand{\Sright}{S_{\mathrm{right}}}
\newcommand{\Smiddle}{S_{\mathrm{middle}}}
\newcommand{\Sleft}{S_{\mathrm{left}}}
\newcommand{\psucc}{p_{\mathrm{succ}}}
\newcommand{\gap}{\mathrm{gap}}

\newcommand{\norm}[1]{\|#1\|}
\newcommand{\abs}[1]{\left|#1\right|}

\DeclareMathOperator{\coin}{coin}
\DeclareMathOperator{\AEst}{\mathsf{AE}}
\DeclareMathOperator{\GAE}{\mathsf{GAE}}
\DeclareMathOperator{\Amplify}{\mathsf{Amplify}}
\DeclareMathOperator{\Estimate}{\mathsf{Estimate}}
\DeclareMathOperator{\Locate}{\mathsf{Locate}}
\DeclareMathOperator{\Shrink}{\mathsf{Shrink}}
\DeclareMathOperator{\BestArm}{\mathsf{BestArm}}
\DeclareMathOperator{\NOT}{\mathsf{NOT}}
\DeclareMathOperator{\SWAP}{\mathsf{SWAP}}
\DeclareMathOperator{\poly}{poly}
\DeclarePairedDelimiter\ceil{\lceil}{\rceil}

\title{Quantum Exploration Algorithms for Multi-Armed Bandits}
\author{
  Daochen Wang\thanks{Equal contribution.}\textsuperscript{\rm 1,2}\quad Xuchen You$^{*}$\textsuperscript{\rm 1,3}\quad Tongyang Li\thanks{Corresponding author. Email: tongyang@mit.edu}\textsuperscript{\rm 1,3,4}\quad Andrew M. Childs\textsuperscript{\rm 1,3}\\
}
\affiliations{
    \textsuperscript{\rm 1}Joint Center for Quantum Information and Computer Science, University of Maryland \\
    \textsuperscript{\rm 2}Department of Mathematics, University of Maryland \\
    \textsuperscript{\rm 3}Department of Computer Science and Institute for Advanced Computer Studies, University of Maryland \\
    \textsuperscript{\rm 4}Center for Theoretical Physics, Massachusetts Institute of Technology \\
    \{daochen,xyou,amchilds\}@umd.edu, tongyang@mit.edu
}

\begin{document}

\maketitle

\begin{abstract}
Identifying the best arm of a multi-armed bandit is a central problem in bandit optimization. We study a quantum computational version of this problem with coherent oracle access to states encoding the reward probabilities of each arm as quantum amplitudes. Specifically, we show that we can find the best arm with fixed confidence using $\tilde{O}\bigl(\sqrt{\sum_{i=2}^n\Delta^{\smash{-2}}_i}\bigr)$ quantum queries, where $\Delta_{i}$ represents the difference between the mean reward of the best arm and the $i^\text{th}$-best arm. This algorithm, based on variable-time amplitude amplification and estimation, gives a quadratic speedup compared to the best possible classical result. We also prove a matching quantum lower bound (up to poly-logarithmic factors).
\end{abstract}

%%%%%%%%%%%%%%%%%%%%%%%%%%%%%%%%%%%%%%%%%%%%%%%%%%%%%%%%%%%%%%%%%%%%%%%%%%%%%%

\section{Introduction}
The multi-armed bandit (MAB) model is one of the most fundamental settings in reinforcement learning. This simple scenario captures crucial issues such as the tradeoff between exploration and exploitation. Furthermore, it has wide applications to areas including operations research, mechanism design, and statistics.

A basic challenge about multi-armed bandits is the problem of \emph{best-arm identification}, where the goal is to efficiently identify the arm with the largest expected reward. This problem captures a common difficulty in practical scenarios, where at unit cost, only partial information about the system of interest can be obtained. A real-world example is a recommendation system, where the goal is to find appealing items for users. For each recommendation, only feedback on the recommended item is obtained. In the context of machine learning, best-arm identification can be viewed as a high-level abstraction and core component of active learning, where the goal is to minimize the uncertainty of an underlying concept, and each step only reveals the label of the data point being queried.

Quantum computing is a promising technology with potential applications to diverse areas including cryptanalysis, optimization, and simulation of quantum physics. Quantum computing devices have recently been demonstrated to experimentally outperform classical computers on a specific sampling task~\cite{arute2019supremacy}. While noise limits the current practical usefulness of quantum computers, they can in principle be made fault tolerant and thus capable of executing a wide variety of algorithms. It is therefore of significant interest to understand quantum algorithms from a theoretical perspective to anticipate future applications. In particular, there has been increasing interest in \emph{quantum machine learning} (see for example the surveys by~\citealt{biamonte2017quantum,schuld2015introduction,arunachalam2017guest,dunjko2018machine}). In this paper, we study best-arm identification in multi-armed bandits, establishing quantum speedup.

%================================================================
\paragraph{Problem setup.}
We work in a standard multi-armed bandit setting~\cite{pac_bandits_evendar_mansour} in which the MAB has $n$ arms, where arm $i \in [n] \coloneqq \{1,\ldots,n\}$ is a Bernoulli random variable taking value $1$ with probability $p_i$ and value $0$ with probability $1-p_i$. Each arm can therefore be regarded as a coin with \emph{bias} $p_i$. As our algorithms and lower bounds are symmetric with respect to the arms, we assume without loss of generality that $p_1\geq\cdots\geq p_n$, and denote $\Delta_i \coloneqq p_1-p_i$ for all $i\in\{2,\ldots,n\}$. We further assume that $p_1>p_2$, i.e., the best arm is unique. Given a parameter $\delta\in(0,1)$, our goal is to use as few queries as possible to determine the best arm with probability $\geq1-\delta$. This is known as the \emph{fixed-confidence setting}. We primarily characterize complexity in terms of the parameter
\begin{equation}
    H \coloneqq \,{\sum_{i=2}^n\ \frac{1}{\Delta^2_i}}
\end{equation}
which arises in the analysis of classical MAB algorithms (as discussed below).

We consider a quantum version of best-arm identification in which we can access the arms \emph{coherently}. This means we have access to a quantum oracle $\mathcal{O}$ that acts as
\begin{equation}\label{eq:quantum-bandit-defn}
\begin{aligned}
\mathcal{O}&\colon \ket{i}_I\ket{0}_B\ket{0}_J\\
&\quad\mapsto \ket{i}_I(\sqrt{\vphantom{1-}p_i}\ket{1}_B\ket{v_i}_J+\sqrt{1-p_i}\ket{0}_B\ket{u_i}_J),
\end{aligned}
\end{equation}
 where $\ket{v_i}$ and $\ket{u_i}$ are arbitrary states, for all $i\in[n]$. We have used standard Dirac notation which we review in the Preliminaries section. Register $I$ is the ``index'' register with $n$ states that correspond to the $n$ arms. Register $B$ is the single-qubit ``bandit'' register with two states, $\ket{1}$ corresponding to a reward and $\ket{0}$ corresponding to no reward. Register $J$ is a multi-qubit ``junk'' register. For convenience, we omit register labels when this causes no confusion. Compared to pulling an arm classically---which can be implemented by measuring the bandit register---the quantum oracle allows access to different arms in superposition, a necessary feature for quantum speedup. In real-world applications, we usually have junk when instantiating our oracle (see below). When deriving our results however, we will assume there is no junk (i.e., we set $\ket{v_i}=\ket{u_i}=1$ for all $i\in[n]$ in \eq{quantum-bandit-defn}). This is without loss of generality as the algorithm we construct is insensitive to junk.

Previous work on quantum algorithms for clustering~\cite{kerenidis2019qmeans, wiebe2015quantum} and reinforcement learning~\cite{dunjko2016quantum,dunjko2018machine} has discussed how to instantiate $\mathcal{O}$. In clustering, $\mathcal{O}$ is created using the $\SWAP$ test where for each $i$, $p_i$ encodes the distance between some fixed vector and the $i^\text{th}$ vector in some collection. Our algorithm can be used to speed up the algorithms of~\citet{kerenidis2019qmeans} and~\citet{wiebe2015quantum}. In reinforcement learning, $\mathcal{O}$ naturally appears in stochastic agent environments; for instance, $\mathcal{O}$ can be viewed as a special case of the oracle in~\citet{dunjko2016quantum} for a Markov decision problem (MDP) of epoch length $1$ and state set $\{0,1\}$, where the goal of the agent is to reach the state $1$.

As a concrete example, consider a classical Monte Carlo strategy\footnote{This is Monte Carlo tree search without tree expansion.}: at a given position, evaluate the quality of a next move $i$ by uniformly randomly playing out games $x\in X(i)$, where $X(i)$ is the set of valid games from move $i$ onwards, and querying a computer program $f$ that computes a bit $f(i,x)\in\{0,1\}$ indicating if game $x$ is won $(1)$ or lost $(0)$.
In the classical case, we obtain one sample of win or loss using one query to $f$. In the quantum case, we can also instantiate one query to the quantum oracle in Eq.~\eqref{eq:quantum-bandit-defn} using just one query to $f$. To do this, we apply the circuit for $f$, made reversible in the usual way~\cite[Sec.~1.4.1]{nielsen2000book}, on the quantum state corresponding to uniformly random play as follows:
\begin{equation}
\begin{aligned}
&\ket{i}\ket{0}\frac{1}{\sqrt{|X(i)|}}\sum_{x\in X(i)}\ket{x} \\
\overset{f}{\mapsto}&
\ket{i}\sum_{x\in X(i)}\frac{1}{\sqrt{|X(i)|}}\ket{f(i,x)}\ket{x} \\
=&\ket{i}(\sqrt{\vphantom{1-}p_i}\ket{1}\ket{u_i} + \sqrt{1-p_i}\ket{0}\ket{v_i}),
\end{aligned}
\end{equation}
where $\ket{u_i}$ and $\ket{v_i}$ are some states, and $p_i$ is the empirical probability that move $i$ leads to a win. Our quantum algorithm then uses quadratically fewer calls to $f$ compared with classical Monte Carlo search to find the best next move.

We stress that we do not need to know the $p_i$s to instantiate the quantum oracle above. We also remark that our algorithm does not apply to every MAB situation. For example, in clinical trials to identify the best drug, we cannot instantiate the quantum oracle because human participants, unlike computer programs, cannot be queried in superposition.

Our algorithm can also be adapted to work when the reward distributions are promised to have bounded variance (for example, if they are sub-Gaussian). The adaptation essentially follows by replacing amplitude estimation (introduced in the Preliminaries section) with quantum mean estimation~\citep{montanaro2015quantum}, which works on any distribution with bounded variance. We remark that the situation is different for the other main type of bandits: adversarial bandits. Studies on adversarial bandits are mainly focused on regret minimization and a quantum analogue first requires a proper notion of regret which we are unsure how to even define.

%================================================================
\paragraph{Contributions.}
In this paper, we give a comprehensive study of best-arm identification using quantum algorithms. Specifically, we obtain the following main result:
\begin{Theorem}\label{thm:main-confidence}
Given a multi-armed bandit oracle $\mathcal{O}$ and confidence parameter $\delta\in(0,1)$, there exists a quantum algorithm that, with probability $\geq1-\delta$, outputs the best arm using
$\tilde{O}\bigl(\sqrt{H}\bigr)$
queries to $\mathcal{O}$. Moreover, this query complexity is optimal up to poly-logarithmic factors in $n$,  $\delta$, and $\Delta_2$.
\end{Theorem}

This represents a quadratic quantum speedup over what is possible classically. The speedup essentially derives from Grover's search algorithm~\cite{grover1996fast}, where a marker oracle is used to approximately ``rotate'' a uniform initial state to the marked state. One way to understand the quadratic speedup is to observe that each rotation step, making one query to the oracle, increases the amplitude of the marked state by $\Omega(1/\sqrt{n})$. This is possible since quantum computation linearly manipulates amplitudes, which are square roots of probabilities.

However, to establish \thm{main-confidence} we use more sophisticated machinery that extends Grover's algorithm, namely variable-time amplitude amplification (VTAA)~\cite{ambainis2012VTAA,childs2015quantum} and estimation (VTAE)~\cite{chakraborty2018power}. We apply VTAA and VTAE on a variable-time quantum algorithm $\A$ that we construct. $\A$ outputs a state with labeled ``good'' and ``bad'' parts. Using that label, VTAA removes the bad part so that only the good part remains, and VTAE estimates the proportion of the good part. In our application, the good part is eventually the best-arm state.

We emphasize that our quantum algorithm, like classical ones~\cite{pac_bandits_evendar_mansour, gabillon2012best, jamieson2014lil, karnin2013almost, mannor2004sample}, does not require any prior knowledge about the $p_i$s.

Given knowledge of $p_1$ and $p_2$, our quantum algorithm is conceptually related to the classical successive elimination (SE) algorithm~\cite{pac_bandits_evendar_mansour}. Namely, we use that knowledge to help eliminate sub-optimal arms $i$ by checking whether $p_i < (p_1+p_2)/2$, say. The quantum quadratic speedup arises because we can check this ``in superposition'' across the different arms. For intuition only, checking in superposition can be thought of as a form of checking in parallel. We stress however that while it does not make sense to compare the parallel (classical) sample complexity of best-arm identification with its usual (classical) sample complexity, it does makes sense to compare the latter with the quantum query complexity. We also stress that the similarity of our quantum algorithm to SE, given knowledge of $p_1$ and $p_2$, ends at the conceptual level. Technically, our algorithm makes the SE concept work by first marking all sub-optimal arms and then rotating towards the unmarked best arm in quantum state space via a careful application of VTAA. This has no classical analogue.

It is classically easy to remove any assumed knowledge of $p_1$ and $p_2$ because classical samples from a multi-armed bandit contain information about their values. Quantumly however, we cannot simply ask our quantum multi-armed bandit to supply \emph{classical} samples as that would prevent interference, eliminating any quantum speedup. Therefore, we need to do something conceptually different in the quantum case. We construct another quantum algorithm whose goal is to estimate both $p_1$ and $p_2$ to precision $\Theta(\Delta_2)$ using $\tilde{O}(\sqrt{H})$ quantum queries. For a given test point $l$, VTAE (roughly) gives us the ability to \emph{count} the number of arms $i$ with $p_i>l$, and thus allows us to perform binary search to find $p_1$ and $p_2$.

%================================================================
\paragraph{Related work.}
Classically, a naive algorithm for best-arm identification is to simply sample each arm the same number of times and output the arm with the best empirical bias~\cite{pac_bandits_evendar_mansour}. This algorithm has complexity $O(\frac{n}{\Delta_2^2}\log(\frac{n}{\delta}))$ but is sub-optimal for most multi-armed bandit instances. Therefore, classical research on best-arm identification~\cite{pac_bandits_evendar_mansour, gabillon2012best, jamieson2014lil, karnin2013almost, mannor2004sample} has primarily focused on proving bounds of the form $\tilde{O}(H)$ (recall that $H \coloneqq \sum_{i=2}^n\frac{1}{\Delta^2_i}$), which can be shown to be almost tight for every instance. The first work to provide an algorithm with such complexity is~\citet{pac_bandits_evendar_mansour}, giving $O(H\log(\frac{n}{\delta}) + \sum_{i=2}^n\Delta^{-2}_i\log(\Delta^{-1}_i))$. This was further improved to $O\bigl(H\log(\frac{1}{\delta}) + \sum_{i=2}^n\Delta^{-2}_i\log\log(\Delta^{-1}_i)\bigr)$
by~\citet{gabillon2012best, jamieson2014lil, karnin2013almost}, which is almost optimal except for the additive term of $\sum_{i=2}^n\Delta^{-2}_i\log\log(\Delta^{-1}_i)$~\cite{mannor2004sample}. More recent work~\cite{chen2015optimal, chen2016towards} has focused on bringing down even this additive term by tightening both the upper and lower bounds, leaving behind a gap only of the order $\sum_{i=2}^n\Delta_i^{-2}\log\log(\min\{n,\Delta^{-1}_i\})$.

Prior work on quantum machine learning has focused primarily on supervised~\cite{lloyd2013quantum,lloyd2013supervised,rebentrost2014QSVM,li2019classification} and unsupervised learning~\cite{lloyd2013supervised,wiebe2015quantum,amin2018quantum,kerenidis2019qmeans}. \citet{dunjko2017advances,dunjko2017exponential,jerbi2019framework} gave quantum algorithms for general reinforcement learning with provable guarantees, but do not consider the best-arm identification problem. The only directly comparable previous work on quantum algorithms for best-arm identification that we are aware of are~\citet{casale2020quantum} and~\citet{wiebe2015quantum}.\footnote{\citet{wiebe2015quantum} is not framed as solving best-arm identification, but is partly concerned with this problem.} By applying Grover's algorithm,~\citet{casale2020quantum} shows that quantum computers can find the best arm with confidence $p_1/ \sum_{i=1}^n p_i$ quadratically faster than classical ones. However,~\citet{casale2020quantum} does not show how to find the best arm with a given \emph{fixed} confidence, which is the standard requirement. In fact, there is a relatively simple quantum algorithm, analogous to the naive classical algorithm, that can achieve arbitrary confidence with quadratic speedup in terms of $n/\Delta_2^2$. This algorithm, which appears in Fig.~3 of~\citet{wiebe2015quantum}, works by using the quantum minimum finding of~\citet{durr1996quantum} on top of quantum amplitude estimation~\cite{amplitude_estimation}. As in the classical case, we show that this simple quantum algorithm is suboptimal for most multi-armed bandit instances. Specifically, we show that a quantum algorithm can achieve quadratic speedup in terms of the parameter $H$.

%%%%%%%%%%%%%%%%%%%%%%%%%%%%%%%%%%%%%%%%%%%%%%%%%%%%%%%%%%%%%%%%%%%%%%%%%%%%%%

\section{Preliminaries}
\paragraph{Definitions and notations.}
    Quantum computing is naturally formulated in terms of linear algebra.
    An $n$-dimensional \emph{quantum state} is a unit vector in the complex Hilbert space $\C^{n}$, i.e., $\vec{x}=(x_{1},\ldots,x_{n})\trans$ such that $\sum_{i=1}^{n}|x_{i}|^{2}=1$. Such a column vector $\vec{x}$ is written in \emph{Dirac notation} as $\ket{x}$ and called a ``ket''. The complex conjugate transpose of $\ket{x}$ is written $\bra{x}$ and called a ``bra'', i.e., $\bra{x} \coloneqq \vec{x}^\dagger$. The reason for the names is because the combination of a bra and a ket is a inner product bracket: $\braket{x|y}\coloneqq \bra{x}\ket{y} = \vec{x}^{\dagger}\vec{y} =\langle x,y\rangle\in \mathbb{C}$.

    The \emph{computational basis} of $\C^{n}$ is the set of vectors $\{\vec{e}_{1},\ldots,\vec{e}_{n}\}$, where $\vec{e}_{i}=(0,\ldots,1,\ldots,0)\trans$ is a one-hot column vector with $1$ in the $i^{\text{th}}$ coordinate. In Dirac notation, it is common to reserve symbols $\ket{i} \coloneqq \vec{e}_{i}$ and $\bra{i} \coloneqq \vec{e}_{i}^\dagger = \vec{e}_{i}\trans$. Then, for example, $\ket{x} = \sum_{i=1}^nx_i\ket{i}$ and $\bra{x} = \sum_{i=1}^nx_i^*\bra{i}$.

    The \emph{tensor product} of quantum states is their Kronecker product: if $\ket{x}\in\C^{n_{1}}$ and $\ket{y}\in\C^{n_{2}}$, then
\begin{align}
\ket{x}\ket{y} &\coloneqq \ket{x}\otimes\ket{y} \\
&\coloneqq (x_{1}y_{1},x_{1}y_{2},\ldots,x_{n_{1}}y_{n_{2}})\trans\in\C^{n_{1}}\otimes\C^{n_{2}}.
\end{align}

A quantum algorithm is a sequence of unitary matrices, i.e., a linear transformation $U$ such that $U^{\dagger}=U^{-1}$.

For any $p \in [0,1]$, we define the  \emph{coin state} in $\C^2$ as
\begin{align}
\ket{\coin{p}} \coloneqq \sqrt{\vphantom{1-}p}\ket{1} + \sqrt{1-p}\ket{0} = (\sqrt{1-p}, \sqrt{\vphantom{1-}p})\trans.
\end{align}
Measuring $\ket{\coin{p}}$ in the computational basis gives $1$ with probability $p$, hence the name.

\paragraph{Quantum multi-arm bandit oracle.} Recall the quantum multi-armed bandit oracle defined in \eq{quantum-bandit-defn}. The arms are accessed in \emph{superposition} by applying the unitary oracle $\mathcal{O}$ on a state $\ket{x}_I\ket{0}_B$ in the joint register of $I$ and $B$. This results in the output quantum state
\begin{align}\label{eq:quantum-mab}
\mathcal{O}\ket{x}_{I}\ket{0}_{B}=\sum_{i=1}^{n}x_{i}\ket{i}_{I}\ket{\coin{p_i}}_B
\end{align}
(recall that we assume there is no junk).
A classical pull of the $i$-th arm can be simulated by choosing $\ket{x}_I = \ket{i}_I$ with $\ket{i}_I\ket{\coin{p_i}}_B$ as the output, and then measuring register $B$ to observe $1$ with probability $p_i$.

In this paper, we mainly focus on \emph{quantum query complexity}, which is defined as the total number of oracle queries. If we have an efficient quantum algorithm for an explicit computational problem in the query complexity setting, then if we are given an explicit circuit realizing the black-box transformation, we will have an efficient quantum algorithm for the problem.

\paragraph{Amplitude amplification and estimation.} Our quantum speed-up can be traced back to \emph{amplitude amplification and estimation}~\cite{amplitude_estimation}. For a classical randomized algorithm for a search problem that returns a correct solution $y$ with probability $p_\mathrm{succ}$, the success probability can be amplified to a constant by $O(1/p_\mathrm{succ})$ repetitions. Let $\A$ be a quantum procedure that outputs a quantum state $\sqrt{p_\mathrm{succ}}\ket{1}\ket{y} + \sqrt{1-p_\mathrm{succ}}\ket{0}\ket{y^\prime}$ for some arbitrary quantum state $\ket{y^\prime}$. Measuring the output state yields the solution $y$ with probability $p_\mathrm{succ}$ just like a classical randomized algorithm. \citet{amplitude_estimation} provided an amplitude amplification procedure that amplifies the amplitude of $\ket{1}\ket{y}$ to a constant with $O(1/\sqrt{p_\mathrm{succ}})$ queries to the quantum procedure $\A$. This effectively provides a randomized algorithm with constant success probability with query complexity $O(t/\sqrt{p_\mathrm{succ}})$ if $\A$ makes $t$ queries to the oracle. The same speed-up can be achieved for the closely related task of estimating $p_\mathrm{succ}$ with \emph{amplitude estimation}.

Amplitude amplification and estimation originates from \emph{Grover's search algorithm.}~\cite{grover1996fast}. The formal statements of Grover's algorithm and amplitude amplification and estimation are postponed to the start of the appendix. We refer the interested reader to the book~\citet{nielsen2000book} on quantum computing for a detailed introduction to basic definitions (Section 3), Grover's algorithm and amplitude amplification (Section 6), and related topics.

\paragraph{Variable-time amplitude amplification and estimation.} \emph{Variable-time amplitude amplification} (VTAA) and \emph{estimation} (VTAE) are procedures that apply on top of so-called variable-time quantum algorithms that may stop at different (variable) time steps with certain probabilities. More precisely, for $t = (t_1, t_2, \cdots, t_m)\in\mathbb{R}^m$ and $w = (w_1, w_2, \cdots, w_m)\in\mathbb{R}^m$, a $(t,w)$-variable-time algorithm $\A$ is one that can be divided into $m$ steps (i.e., $\A = \A_m\cdots\A_1$) where $t_j$ is the query complexity of $\A_j\cdots\A_1$ and $w_j$ is the probability of stopping at step $j$. We have:
\begin{Theorem}[Informal: Variable-time amplitude amplification and estimation--\citealp{ambainis2012VTAA,childs2015quantum, chakraborty2018power}]
\label{thm:vtaa_vtae}
Given a $(t,w)$-variable-time quantum algorithm $\A = \A_m \cdots \A_1$ with success probability $p_\mathrm{succ}$, there exists a quantum algorithm $\A^\prime$ that uses $O(Q)$ queries to output the solution with probability $\geq \frac{1}{2}$, where
\begin{equation}~\label{eq:vtaa_complexity_main}
Q \coloneqq t_m \log(t_m) + \frac{t_{\mathrm{avg}}}{\sqrt{\psucc}}\log(t_m).
\end{equation}
with $t_\mathrm{avg} \coloneqq \sqrt{\textstyle\sum_{j=1}^m w_j t_j^2}$ being the root-mean-square average query complexity of $\A$.

There also exists a quantum algorithm that uses $O(\frac{Q}{\epsilon}\log^2 (t_m)\log\log(\frac{t_m}{\delta}))$ queries to estimate $\psucc$ with multiplicative error $\epsilon$ with probability $\geq1-\delta$.
\end{Theorem}

For comparison, recall that applying amplitude amplification and estimation procedures on general quantum algorithms requires $O(t_m/\sqrt{p_\mathrm{succ}})$ queries. See the first section of the appendix for a rigorous definition of variable-time algorithms and formal statements of the query complexities of variable-time amplitude amplification and estimation.

\section{Fast Quantum Algorithm For Best-arm Identification}\label{sec:BAI}

In this section, we construct a quantum algorithm for best-arm identification and analyze its performance. Specifically:

\begin{Theorem}\label{thm:quantum-BAI-confidence}
Given a multi-armed bandit oracle $\mathcal{O}$ and confidence parameter $\delta\in(0,1)$, there exists a quantum algorithm that outputs the best arm with probability $\geq 1-\delta$ using $\tilde{O}(\sqrt{H})$ queries to $\mathcal{O}$.
\end{Theorem}

Throughout this section, the oracle $\mathcal{O}$ is fixed, so we may omit explicit reference to it. All $\log$s have base $2$.

There are essentially two steps in our construction. In the first step, we construct two subroutines $\Amplify$ and $\Estimate$ using VTAA and VTAE, respectively, on a variable-time quantum algorithm $\mathcal{A}$. Roughly speaking, given $l\in [0,1]$, $\Amplify$ outputs an arm index $i$ randomly chosen from those $i$ with $p_i>l$ while $\Estimate$ counts the number of such $i$s. This means that if we knew the values of $p_1$ and $p_2$, we could take $l$ to be $(p_1+p_2)/2$, then $\Amplify$ would output the best arm. But we can use $\Estimate$ in a binary search procedure to estimate $p_1$ and $p_2$. This is exactly what we do in the second step and so we are done.

We now discuss the construction more precisely. $\Amplify$ and $\Estimate$ actually use two thresholds $l_2$, $l_1\in [0,1]$ with $l_2<l_1$ instead of a single threshold $l$. In the first step, we construct a variable-time quantum algorithm denoted $\A$ (\algo{vtalgo}) that is initialized in a uniform superposition state $\ket{u} \coloneqq \frac{1}{\sqrt{n}}\sum_{i \in [n]}\ket{i}$ (since initially we have no information about which arm is the best). Given an input interval $I=[l_2,l_1]$, $\A$ ``flags'' arm indices in $\Sright' \coloneqq \{i \in [n]: p_i \geq l_1\}$ with a bit $f=1$ and those in
$\Sleft' \coloneqq \{i \in [n]: p_i \leq l_2\}$ with a bit $f=0$. The flag bit $f$ is written to a separate flag register $F$, so that the state (approximately) becomes $\frac{1}{\sqrt{n}}\bigl(\sum_{i\in \Sright'}\ket{i}\ket{1}_F + \sum_{i\in \Sleft'}\ket{i}\ket{0}_F + \sum_{i\in \Smiddle'}\ket{i}\ket{\psi_i}_F\bigr)$ for some states $\ket{\psi_i}\in\C^2$, where $\Smiddle' \coloneqq [n] - (\Sleft'\cup\Sright') = \{i \in [n]: l_2<p_i<l_1\}$. The flag bit $f$ stored in the $F$ register indicates whether VTAA (resp.\ VTAE), when applied on $\A$, should ($f=1$) or should not ($f=0$) amplify (resp.\ estimate) that part of the state. We then apply VTAA and VTAE on $\A$ to construct $\Amplify$ and $\Estimate$, respectively. $\Amplify$ produces a uniform superposition of all those $i$s with $F$ register in $\ket{1}$, i.e., it amplifies such $i$s relative to the others. $\Estimate$ counts the number of such $i$s. More precisely, $\Estimate$ (approximately) counts the number of indices in $\Sright'$, as their $F$ register is in $\ket{1}$, plus some (unknown) fraction of indices in $\Smiddle'$ as dictated by the fraction of $\ket{1}$ in the (unknown) states $\ket{\psi_i}$.

In the second step, we use $\Estimate$ as a subroutine in $\Locate$ (\algo{locate}) to find a interval $[l_2,l_1]$ such that $p_2 < l_2 < l_1 < p_1$  and that $|l_1 - l_2| \geq \Delta_2/4$. Then, running $\Amplify$ with these $l_2,l_1$ in $\BestArm$ (\algo{bestarm}) gives the state $\ket{1}$ containing the best-arm index because only $p_1$ is to the right of $l_2$. $\Locate$ is a type of binary search that counts the number of indices in $\Sright'$ using $\Estimate$. There is a technical difficulty here because $\Estimate$ actually counts the number of indices in $\Sright'$ plus some fraction of indices in $\Smiddle'$. Trying to fix this by simply setting $l_2=l_1$, so that $\Smiddle' = \emptyset$, does not work as it would increase the cost of $\Estimate$. We overcome this difficulty via the $\Shrink$ subroutine (\algo{shrink}) of $\Locate$, which employs a technique from recent work on quantum ground state preparation~\cite{lintong}. See \fig{flowchart} for an illustration of the overall structure of the algorithm.
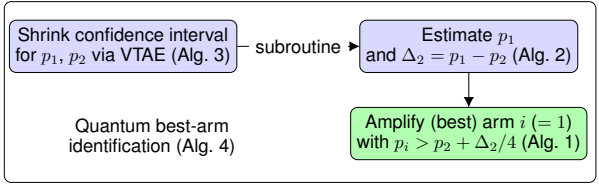
\begin{figure}[htbp]
\centering
\scalebox{0.65}{
{\tikzset{%
  >={Latex[width=2mm,length=2mm]},
            base/.style = {rectangle, rounded corners, draw=black,
                           minimum width=4cm, minimum height=1cm,
                           text centered, font=\sffamily},
       result/.style = {base, fill=green!30},
       new/.style = {base, fill=blue!15},
         old/.style = {base, minimum width=2.5cm, fill=yellow!30},
}
\begin{tikzpicture}[node distance=2cm,
    every node/.style={fill=white, font=\sffamily, scale=1.07}, align=center]
  \draw [rounded corners] (-9.5,-1) rectangle node[pos=0.25]{Quantum  best-arm\\  identification (Alg.~\ref{algo:bestarm})} (2.6,2.7);
  \node (A-known-p)     [result]          {Amplify (best) arm $i$ ($=1$)\\
   with $p_i > p_2+\Delta_2/4$ (Alg.~\ref{algo:vtalgo})};
  \node (locate)     [new, above = 0.7 cm of A-known-p]   {Estimate $p_1$ \\and $\Delta_2 = p_1-p_2$ (Alg.~\ref{algo:locate})};
  \node (shrink)     [new, left = 2.5 cm of locate]   {Shrink confidence interval \\ for $p_1$, $p_2$ via VTAE (Alg.~\ref{algo:shrink})};

  \draw[->]              (locate) -- (A-known-p);
  \draw[->]              (shrink) -- node {subroutine} (locate);
\end{tikzpicture}}}
\caption{Overview of our best-arm identification algorithm.}
\label{fig:flowchart}
\end{figure}

\subsection{$\Amplify$ and $\Estimate$}\label{sec:step_one}
We first construct a variable-time quantum algorithm (\algo{vtalgo}) that we call $\A$ throughout. $\A$ uses the following registers: input register $I$; bandit register $B$; clock register $C=(C_1,\ldots, C_{m+1})$, where each $C_i$ is a qubit; ancillary amplitude estimation register $P=(P_1,\ldots, P_m)$, where each $P_i$ has $O(m)$ qubits; and flag register $F$. We set $m \coloneqq \ceil{\log(1/(l_1-l_2))}+2$ as assigned in \algo{vtalgo}.

$\A$ is indeed a variable-time quantum algorithm according to \defn{variable_time_algorithm}. This is because we can write $\A = \A_{m+1}\A_{m}\cdots\A_{1}\A_{0}$ as a product of $m+2$ sub-algorithms, where $\A_0$ is the initialization step (\lin{vtalgo_initial}), $\A_j$ consists of the operations in iteration $j$ of the for loop (Lines~\ref{lin:vtalgo_forloop1}--\ref{lin:vtalgo_forloop4}) for $j\in [m]$, and $\A_{m+1}$ is the termination step (Lines~\ref{lin:vtalgo_terminate1}--\ref{lin:vtalgo_terminate2}). The state spaces $\H_C$ and $\H_A$ in \defn{variable_time_algorithm} correspond to the state spaces of the $C$ register and the remaining registers of $\A$, respectively. $\A_{m+1}$ ensures that Condition \ref{it:vtqa_final} of \defn{variable_time_algorithm} is satisfied.

\begin{algorithm}[htbp]
\KwInput{Oracle $\mathcal{O}$ as in \eq{quantum-bandit-defn}; $0 < l_2 < l_1 < 1$; approximation parameter $0<\alpha<1$.}

$\Delta \leftarrow l_1-l_2$

$m \leftarrow \ceil{\log \frac{1}{\Delta}}+2$

$a \leftarrow \frac{\alpha}{2mn^{3/2}}$

Initialize state to $\frac{1}{\sqrt{n}}\sum_{i=1}^n{\ket{i}}_I\ket{\coin p_i}_B\ket{0}_C\ket{0}_P\ket{1}_F$\label{lin:vtalgo_initial}

\For{$j = 1, \ldots, m$\label{lin:vtalgo_forloop}}{
    $\epsilon_j \leftarrow 2^{-j}$\label{lin:vtalgo_forloop1}

    \If{register $I$ is in state $\ket{i}$ and registers $C_1,\ldots, C_{j-1}$ are in state $\ket{0}$}{
    Apply $\GAE(\epsilon_j,a; l_1)$ with $\mathcal{O}_{p_i}$ on registers $B$, $C_j$, and $P_j$\label{lin:vtalgo_gae}
    }
    Apply controlled-$\NOT$ gate with control on register $C_j$ and target on register $F$\label{lin:vtalgo_forloop4}
    }
\If{registers $C_1,\ldots, C_{m}$ are in state $\ket{0}$\label{lin:vtalgo_terminate1}}{
    Flip the bit stored in register $C_{m+1}$\label{lin:vtalgo_terminate2}}
\caption{$\A(\mathcal{O},l_2,l_1,\alpha)$}
\label{algo:vtalgo}
\end{algorithm}

With $\Delta := l_1-l_2$ being the length of $[l_2,l_1]$, we define the following three sets that partition $[n]$:
\begin{align}
\Sleft &\coloneqq \{i \in [n]: \,p_i < l_1 - \Delta/2\}\label{eq:Sleft},
\\
\Smiddle &\coloneqq \{i \in [n]: \, l_1-\Delta/2\leq p_i
  <l_1-\Delta/8\}\label{eq:Smiddle},
\\
\Sright &\coloneqq \{i \in [n]: \,p_i \geq  l_1 - \Delta/8\}\label{eq:Sright}.
\end{align}
These sets play the roles of aforementioned $\Sleft'$, $\Smiddle'$, and $\Sright'$. 
They can be regarded as functions of (the input to) $\A$. For later convenience, we also define $S_{\mathrm{lm}}\coloneqq\Sleft \cup \Smiddle$ and $S_{\mathrm{mr}}\coloneqq\Smiddle \cup \Sright$.

\begin{Lemma}[Correctness of $\A$]\label{lem:vtalgo_correctness}
Let $\psucc$ denote the success probability $\A$. Then $\abs{\psucc-\psucc'}\leq \frac{2\alpha}{n}$ where $\psucc' = \frac{1}{n}\bigl(\abs{\Sright}+\sum_{i \in \Smiddle}\abs{\beta_{i,1}}^2\bigr)$ for some $\abs{\beta_{i,1}}^2\in [0,1]$.
\end{Lemma}

At a high level, at iteration $j$, \lin{vtalgo_gae} approximately identifies those $i\in \Sleft$ with $p_i \in [l_1-2\epsilon_j,l_1-\epsilon_j)$ and stops computation on these $i$s by setting their associated $C$ registers to $\ket{1}$. \lin{vtalgo_forloop4} then flags these $i$s by setting their associated $F$ registers to $\ket{0}$, indicating failure. We defer the detailed proof to the supplementary material which is mainly concerned with bounding the error in the aforementioned approximation, as well as the lemma as follows.

\begin{Lemma}[Complexity of $\A$]\label{lem:vtalgo_complexity}
With $\Delta = l_1 - l_2$ being the length of the interval,
we have:
\begin{enumerate}[nosep]
\item The $j^\text{th}$ stopping time $t_j$ of $\A_j\A_{j-1}\cdots \A_{0}$ is of order $\sum_{k=1}^j\frac{1}{\epsilon_k}\log\frac{1}{a} \leq 2^{j+1}\log\frac{1}{a}$. In particular, $t_{m+1}=O(\frac{1}{\Delta}\log\frac{1}{a})$.

\item The average stopping time squared, $t_{\mathrm{avg}}^2$, is of order
\begin{equation}\label{eq:t_avg}
\frac{1}{n}\biggl(\frac{\abs{\Sright}}{\Delta^2} + \sum_{i\in S_{\mathrm{lm}}} \frac{1}{(l_1-p_i)^2}\bigg) \log^2\Bigl(\frac{1}{a}\Bigr).
\end{equation}
\end{enumerate}
\end{Lemma}

Now we fix algorithm $\mathcal{A}$ and its input parameters. We always assume that $\abs{\Sright}\geq1$, which we need for some of the following results to hold. This is without loss of generality as we can always add an artificial arm $0$ with bias $p_0=1$ to the bandit oracle $\mathcal{O}$, as we do in \lin{shrink_append} of \algo{shrink}.

We apply VTAA and VTAE (\thm{vtaa_vtae})\footnote{The state spaces $\H_C$, $\H_F$, and $\H_W$ correspond to the state spaces of the $C$, $F$, and remaining registers of $\A$, respectively.} on our variable-time quantum algorithm $\A$ to prepare the state $\ket{\psi_{\textrm{succ}}}$ and to estimate the probability $\psucc$, respectively. This gives two new algorithms $\Amplify$ and $\Estimate$ with the following performance guarantees.

\begin{Lemma}[Correctness and complexity of $\Amplify(\A, \delta)$, $\Estimate(\A,\epsilon,\delta)$]~\label{lem:vtaa_vtae_on_vtalgo}
Let $\A = \A(\mathcal{O},l_2,l_1,0.01\delta)$.
Then $\Amplify(\A, \delta)$ uses $O(Q)$ queries to output an index $i\in S_{\mathrm{mr}}$ with probability $\geq 1-\delta$, and $\Estimate(\A, \epsilon, \delta)$ uses $O(Q/\epsilon)$ queries to output an estimate $r$ of $\psucc'$ (defined in \lem{vtalgo_correctness}) such that
\begin{equation}~\label{eq:vtae_estimate_quality}
(1-\epsilon)\Bigl(\psucc'-\frac{0.1}{n}\Bigr) < r < (1+\epsilon)\Bigl(\psucc' + \frac{0.1}{n}\Bigr)
\end{equation}
with probability $\geq 1-\delta$,
where $Q$ is
\begin{equation}\label{eq:coarse_vtaa_query_general}
\biggl(\frac{1}{\Delta^2} + \frac{1}{\abs{\Sright}}\sum_{S_{\mathrm{lm}}} \frac{1}{(l_1-p_i)^2}\biggr)\,
\poly\Bigl(\log\Bigl(\frac{n}{\delta \Delta}\Bigr)\Bigr),
\end{equation}
where $\Delta = l_1 - l_2$.
\end{Lemma}

This lemma follows by applying \lem{vtalgo_correctness} and \lem{vtalgo_complexity} to \thm{vtaa_vtae}. The proof detail is given in the appendices.

 \subsection{Quantum algorithm for best-arm identification}\label{sec:step_three}
In this subsection, we use $\Amplify$ and $\Estimate$ to construct three algorithms (Algorithms~\ref{algo:locate}--\ref{algo:bestarm}) that work together to identify the best arm following the outline that we described at the beginning of this section.

\begin{algorithm}[htbp]
\KwInput{Oracle $\mathcal{O}$ as in \eq{quantum-bandit-defn}; confidence parameter $0<\delta<1$.}
$I_1, I_2 \leftarrow [0,1]$

$\delta \leftarrow \delta / 8$

\While{$\min I_1 - \max I_2 < 2\abs{I_1}$\label{lin:locate_while}}{
$I_1 \leftarrow \Shrink(\mathcal{O},1, I_1, \delta)$

$I_2 \leftarrow \Shrink(\mathcal{O},2, I_2, \delta)$

$\delta \leftarrow \delta/2$\label{lin:locate_delta_halve}
}
\Return{$I_1, I_2$}
\caption{$\Locate(\mathcal{O},\delta)$}
\label{algo:locate}
\end{algorithm}

\begin{algorithm}[ht]
\KwInput{Oracle $\mathcal{O}$ as in \eq{quantum-bandit-defn}; $k \in \{1,2\}$;
  interval $I=[a,b]$; confidence parameter $0<\delta<1$.}

$\epsilon\leftarrow (b-a)/5$

$\delta\leftarrow \delta/2$

Append arm $i=0$ with bias $p_0=1$ to $\mathcal{O}$; call the resulting oracle $\mathcal{O}'$\label{lin:shrink_append}

Construct variable-time quantum algorithms $\A_1,\A_2$:

\quad $\A_1 \leftarrow \A(\mathcal{O}', l_2 = a+\epsilon, l_1 = a+3\epsilon, 0.01\delta)$\label{lin:shrink_vtalgo1}

\quad $\A_2 \leftarrow \A(\mathcal{O}', l_2 = a+2\epsilon, l_1 = a+4\epsilon, 0.01\delta)$\label{lin:shrink_vtalgo2}

$r_1 \leftarrow \Estimate(\A_1, \epsilon = 0.1, \delta)$\label{lin:shrink_estimate1}

$r_2 \leftarrow \Estimate(\A_2, \epsilon = 0.1, \delta)$\label{lin:shrink_estimate2}

$B_1\leftarrow \mathds{1}({r_1 > \frac{k+0.5}{n+1}})$; $B_2\leftarrow \mathds{1}({r_2 > \frac{k+0.5}{n+1}})$

\Switch{$(B_1,B_2)$}
{
\textbf{case} \ $(0,0): I \leftarrow [a, a+3\epsilon]$\label{lin:shrink_switch_case1}

\textbf{case} \ $(0,1): I \leftarrow [a+\epsilon, a+4\epsilon]$

\textbf{case} \ $(1,0): I\leftarrow [a+\epsilon, a+4\epsilon]$

\textbf{case} \ $(1,1): I \leftarrow [a+2\epsilon, a+5\epsilon = b]$\label{lin:shrink_switch_case4}
}
\textbf{return} \ $I$
\caption{$\Shrink(\mathcal{O},k,I,\delta)$}
\label{algo:shrink}
\end{algorithm}

\begin{algorithm}[ht]
\KwInput{Oracle $\mathcal{O}$ as in \eq{quantum-bandit-defn}; confidence parameter $0<\delta<1$.}

$\delta \leftarrow \delta/2$\label{lin:bestarm_delta}

$I_1,I_2 \leftarrow \Locate(\mathcal{O},\delta)$\label{lin:bestarm_locate}

$l_1 \leftarrow  \min I_1$ (left endpoint of $I_1$)\label{lin:bestarm_l1}

$l_2 \leftarrow  \max I_2$ (right endpoint of $I_2$)\label{lin:bestarm_l2}

Construct variable-time quantum algorithm $\A$:

\quad $\A \leftarrow \A(\mathcal{O}, l_2, l_1, 0.01\delta)$
\label{lin:bestarm_vtalgo}

$i \leftarrow \Amplify(\A, \delta)$

\Return{i}
\caption{$\BestArm(\mathcal{O},\delta)$}
\label{algo:bestarm}
\end{algorithm}

We state the correctness and complexities of $\Amplify$ and $\Estimate$ as follows:
\begin{Lemma}[Correctness and complexity of \algo{locate}]\label{lem:locate}
Fix a confidence parameter $0<\delta<1$. Then the event $E = \{p_1\in I_1 \text{ and } p_2\in I_2 \text{ in all iterations of the while loop}\}$ holds with probability $\geq1-\delta$. When $E$ holds, \algo{locate} also satisfies the following for both $k\in\{1,2\}$:
\begin{enumerate}[nosep]
    \item its while loop (\lin{locate_while}) breaks at or before the end of iteration  $\ceil{\log_{5/3}(\frac{1}{\Delta_2})}+3$ and then returns $I_k$ with $p_k\in I_k$ and $\min I_1 - \max I_2 \geq 2\abs{I_1}$; during the while loop, we always have $\abs{I_1} = \abs{I_2} \geq \Delta_2/8$; and
    \item it uses
    $O\bigl(\sqrt{H} \, \poly\bigl(\log\bigl(\frac{n}{\delta \,\Delta_2}\bigr)\bigr)\bigr)$ queries.
\end{enumerate}
\end{Lemma}

\begin{Lemma}[Correctness and complexity of \algo{shrink}]
\label{lem:shrink}
Fix $k\in\{1,2\}$, an interval $I = [a,b]$, and a confidence parameter $0<\delta<1$. Suppose that $p_k\in I$ and $\abs{I} \geq \Delta_2/8$. Then \algo{shrink}
\begin{enumerate}[nosep]
\item outputs an interval $J$ with $|J| = \frac{3}{5}\abs{I}$ such that $p_k\in J$ with probability $\geq 1-\delta$, and
\item uses $O\bigl(\sqrt{H} \, \poly\bigl(\log\bigl(\frac{n}{\delta \,\Delta_2}\bigr)\bigr)\bigr)$ queries.
\end{enumerate}
\end{Lemma}

The proofs of \lem{locate} and \lem{shrink} appear in the supplementary material.

The following theorem is equivalent to \thm{quantum-BAI-confidence}.

\begin{Theorem}[Correctness and complexity of \algo{bestarm}]\label{thm:bestarm}
Fix a confidence parameter $0<\delta<1$. Then, with probability $\geq1-\delta$, \algo{bestarm}
\begin{enumerate}[nosep]
\item outputs the best arm, and
\item uses $O\bigl(\sqrt{H}\,\poly\bigl(\log\bigl(\frac{n}{\delta \,\Delta_2}\bigr)\bigr)\bigr)$ queries.
\end{enumerate}
\end{Theorem}

\begin{proof}
Note that $\delta$ is halved at the beginning, on \lin{bestarm_delta}. For the first claim, we know from the first claim of \lem{locate} that, with probability $\geq1-\delta/2$, the two intervals $I_k$ assigned in \lin{bestarm_locate} have $\min I_1 - \max I_2 \geq 2\abs{I_1} \geq \Delta_2/4$
and $p_k\in I_k$. Assuming this holds, we have $p_2<l_2<l_2+\Delta_2/4 \leq l_1<p_1$ for the endpoints $l_k$ assigned in Lines~\ref{lin:bestarm_l1} and \ref{lin:bestarm_l2}. This means that the variable-time quantum algorithm $\A$ defined in \lin{bestarm_vtalgo} has $\Sright\cup\Smiddle=\{1\}$, so $\Amplify(\A,\delta/2)$ returns index $1$ with probability $\geq1-\delta/2$. Therefore, the overall probability of \algo{bestarm} returning the best arm is at least $1-\delta$.

The second claim follows immediately from adding the complexity of $\Locate(\mathcal{O},\delta/2)$ (\lem{locate}) and $\Amplify(\A,\delta/2)$ (\lem{vtaa_vtae_on_vtalgo}, using $l_1-l_2\geq\Delta_2/4$).
\end{proof}

By establishing \thm{bestarm}, we have established \thm{quantum-BAI-confidence}, our main claim. %the main claim of \sec{BAI}.
As discussed previously, the main complexity measure of interest in the classical case is $H$, and we see that we get a quadratic speedup in terms of this parameter.

We can see that the poly-logarithmic factor has degree about $6$ from \eq{detailed_vtaa_query_general}, \eq{detailed_vtae_query_general}, and \eq{bestarm_complexity}. It would be interesting to reduce this degree. A more fundamental challenge is to remove the variable $n$ that appears in our log factors. In the classical case, $n$ was already removed from log factors in early work~\cite{pac_bandits_evendar_mansour} by a procedure called ``median elimination''.
However, quantizing the median elimination framework is nontrivial, as the query complexity for outputting the $n/2$ smallest items among $n$ elements is $\Theta(n)$~\cite[Theorem 1]{ambainis2010new}, exceeding our budget of $O(\sqrt{n})$.

As corollaries of our main results in the fixed-confidence setting, we provide results on best-arm identification in the PAC (Probably Approximately Correct) and fixed-budget settings. In the $(\epsilon,\delta)$-PAC setting, the goal is to identify an arm $i$ with $ p_i \geq p_1 - \epsilon$ with probability $\geq1-\delta$. Our best-arm identification algorithm can be modified to work in this setting as well. More precisely, we can modify $\Locate$ (\algo{locate}) by adding a breaking condition to the while loop when $|I_1|$ (or equivalently $|I_2|$) is smaller than $\epsilon$. This gives the following result:

\begin{Corollary}
There is a quantum algorithm that finds an $\epsilon$-optimal arm with query complexity
$O\bigl(\sqrt{\min\{\frac{n}{\epsilon^2},H\}}\cdot \poly\bigl(\log\bigl(\frac{n}{\delta \,\Delta_2}\bigr)\bigr)\bigr)$.
\end{Corollary}

Note that our modification means that the $\Amplify$ step in \algo{bestarm} takes an input interval $I$ with $|I| = l_1 - l_2\in[\epsilon/2, \epsilon]$. The correctness and complexity follow directly from \lem{vtalgo_correctness} and \lem{vtaa_vtae_on_vtalgo}. For comparison,~\citet{pac_bandits_evendar_mansour} gave a classical PAC algorithm with complexity $O\bigl(\frac{n}{\epsilon^2}\log\bigl(\frac{n}{\delta}\bigr)\bigr)$, which was later improved to $O\bigl(\sum_{i=1}^n\min\{\epsilon^{-2},\Delta_i^{-2}\}\cdot\log\bigl(\frac{n}{\delta\Delta_2}\bigr)\bigr)$ by~\citet{gabillon2012best}.

In the supplementary material, we also show how to identify the best arm with high probability for a fixed number of total queries (the fixed-budget setting) given knowledge of $H$.

%%%%%%%%%%%%%%%%%%%%%%%%%%%%%%%%%%%%%%%%%%%%%%%%%%%%%%%%%%%%%%%%%%%%%%%%%%%%%%

\section{Quantum lower bound}\label{sec:quantum-lower}
In this section, we describe a lower bound for the quantum best-arm identification problem. Our lower bound shows that the algorithm of \thm{quantum-BAI-confidence} is optimal up to poly-logarithmic factors.

\begin{restatable}{Theorem}{quantumlower}\label{thm:quantum-BAI-confidence-lower}
Let $p\in (0,1/2)$. For any biases $p_i \in [p,1-p]$, any quantum algorithm that identifies the best arm requires $\Omega(\sqrt{H})$ queries to the multi-armed bandit oracle $\mathcal{O}$.
\end{restatable}

To prove this lower bound, we use the quantum adversary method to show quantum hardness of distinguishing $n$ oracles $\mathcal{O}_x$, $x\in [n]$, corresponding to the following $n$ bandits. In the $1^\text{st}$ bandit, we assign bias $p_i$ to arm $i$ for all $i$. In the $x^\text{th}$ bandit for $x\in\{2,\ldots,n\}$, we assign bias $p_1+\eta$ to arm $x$ and $p_i$ to arm $i$ for all $i\neq x$, where $\eta$ is an appropriately chosen parameter. This hard set of bandits is inspired by the proof of a corresponding classical lower bound~\cite[Theorem 5]{mannor2004sample}.

More precisely, for a positive integer $T$, consider an arbitrary $T$-query quantum algorithm that distinguishes the oracles $\mathcal{O}_x$. The main idea of the adversary method is to keep track of certain quantities $s_k\in \mathbb{R}$ where $k\in \{0,1,\dots, T\}$. For each $k$, $s_k$ quantifies how close the states of the quantum algorithm are when it operates using $k$ queries to the different $\mathcal{O}_x$. At the start, when $k=0$, $s_0$ must be large because when no queries have been made, the states must be close. At the end, when $k=T$, $s_T$ must be small because the states are distinguishable by assumption.

The key point is that we can also bound how much $s_k$ can change in one query, that is we can bound the quantities $\abs{s_{k+1}-s_k}$ for each $k$. Of course, this bound immediately gives a lower bound on $T$, the number of queries it takes to go from $s_0$ (large) to $s_T$ (small). To bound $\abs{s_{k+1}-s_k}$, the key point is to bound the distance between oracles, i.e. matrices, $\mathcal{O}_x$ and  $\mathcal{O}_y$ for different $x,y\in [n]$.

We defer the full proof and full description of the quantum adversary method to the supplementary material.

%%%%%%%%%%%%%%%%%%%%%%%%%%%%%%%%%%%%%%%%%%%%%%%%%%%%%%%%%%%%%%%%%%%%%%%%%%%%%%

\section{Conclusions}
In this paper, we propose a quantum algorithm for identifying the best arm of a multi-armed bandit, which gives a quadratic speedup compared to the best possible classical result. We also prove a matching quantum lower bound (up to poly-logarithmic factors).

This work leaves several natural open questions:
\begin{itemize}
\item Can we give fast quantum algorithms for the exploitation of multi-armed bandits? In particular, can we give online algorithms with favorable regret? The quantum hedging algorithm~\cite{hamoudi2020hedging} and the quantum boosting algorithm~\cite{arunachalam2020boosting} might be relevant to this challenge.

\item Can we give fast quantum algorithms for other types of multi-armed bandits, such as contextual bandits or adversarial bandits (e.g., \citealt{beygelzimer2011contextual, agarwal2014taming,auer2002nonstochastic})?

\item Can we give fast quantum algorithms for finding a near-optimal policy of a Markov decision process (MDP)? MDPs are a natural generalization of MABs, where the goal is to maximize the expected reward over sequences of decisions. \citet{pac_bandits_evendar_mansour} gave a reduction from this problem to best-arm identification by viewing the Q-function of each state as a multi-armed bandit.
\end{itemize}

%%%%%%%%%%%%%%%%%%%%%%%%%%%%%%%%%%%%%%%%%%%%%%%%%%%%%%%%%%%%%%%%%%%%%%%%%%%%%%

\section{Ethics Statement}
This work is purely theoretical. Researchers working on theoretical aspects of bandits and quantum computing may immediately benefit from our results. In the long term, once fault-tolerant quantum computers have been built, our results may find practical applications in multi-armed bandit scenarios arising in the real world. As far as we are aware, our work does not have negative ethical impact.

%%%%%%%%%%%%%%%%%%%%%%%%%%%%%%%%%%%%%%%%%%%%%%%%%%%%%%%%%%%%%%%%%%%%%%%%%%%%%%

\section{Acknowledgements}
DW thanks Robin Kothari, Jin-Peng Liu, Yuan Su, and Aarthi Sundaram for helpful discussions. This work received support from the Army Research Office (W911NF-17-1-0433 and W911NF-20-1-0015); the National Science Foundation (CCF-1755800, CCF-1813814, and PHY-1818914); and the U.S.\ Department of Energy, Office of Science, Office of Advanced Scientific Computing Research, Quantum Algorithms Teams and Accelerated Research in Quantum Computing programs. DW and TL were also supported by QISE-NET Triplet Awards (NSF grant DMR-1747426) and TL by an IBM PhD Fellowship.

%%%%%%%%%%%%%%%%%%%%%%%%%%%%%%%%%%%%%%%%%%%%%%%%%%%%%%%%%%%%%%%%%%%%%%%%%%%%%%

\newcommand{\arxiv}[1]{
  \href{https://arxiv.org/abs/#1}{\ttfamily{arXiv:#1}}\?}\newcommand{\arXiv}[1]{
  \href{https://arxiv.org/abs/#1}{\ttfamily{arXiv:#1}}\?}\def\?#1{\if.#1{}\else#1\fi}

%%%%%%%%%%%%%%%%%%%%%%%%%%%%%%%%%%%%%%%%%%%%%%%%%%%%%%%%%%%%%%%%%%%%%%%%%%%%%%
\newpage
\onecolumn
\appendix

\section{Preliminaries on Quantum Algorithms}
\subsection{Grover's search and amplitude amplification and estimation} Our quantum speedup conceptually originates from \emph{Grover's search algorithm}~\cite{grover1996fast}. Consider a function $f_{w}\colon\range{n}\rightarrow\{-1,1\}$ such that $f_{w}(i)=1$ if and only if $i\neq w$, so that $w$ can be viewed as a (unique) marked item. To search for $w$, classically we need $\Omega(n)$ queries to $f_{w}$.
Quantumly, we can use one call of $f_w$ to create an oracle $U_{w}$ such that $U_{w}\ket{i}=\ket{i}$ for all $i\neq w$ and $U_{w}\ket{w}=-\ket{w}$.Now consider the uniform superposition $\ket{u} \coloneqq \frac{1}{\sqrt{n}}\sum_{i\in\range{n}}\ket{i}$ as well as the state $\ket{r} \coloneqq \frac{1}{\sqrt{n-1}}\sum_{i\in\range{n}/\{w\}}\ket{i}$. The angle between $U_{w}\ket{u}$ and $\ket{u}$ is $\theta\coloneqq\arccos(1/n)=\Theta(1/\sqrt{n})$. Note that the unitary $U_{w}$ reflects about $\ket{r}$, and the unitary $U_{u}=2\ket{u}\bra{u}-I$ reflects about $\ket{u}$. If we start with $\ket{u}$, the angle between $U_{w}\ket{u}$ and $U_{u}U_{w}\ket{u}$ is \emph{amplified} to $2\theta$, and in general the angle between $U_{w}\ket{u}$ and $(U_{u}U_{w})^{k}\ket{u}$ is $2k\theta$. It thus suffices to take $k=\Theta(\sqrt{n})$ to find $w$.

This method of alternatively applying two reflections to boost the amplitude for success can be generalized to a technique called \emph{amplitude amplification}. For the case with some unknown number $k\in[n]$ of marked items, there is also a quadratic quantum speedup for estimating $\theta\coloneqq\arccos(k/n)$ via a technique called \emph{amplitude estimation}~\cite{amplitude_estimation}.

In the context of searching, consider a quantum procedure $\A$ that returns a state $\ket{\psi}$ with $t$ oracle queries, such that the overlap between the target state $\ket{w}$ and output state $\ket{\psi}$ is $p_{\mathrm{succ}} \coloneqq |{\braket{w|\psi}}|^2$. By amplitude amplification and estimation~\cite{amplitude_estimation}, $O({t}/{\sqrt{p_{\mathrm{succ}}}})$ oracle queries suffice to amplify the overlap to constant order and to estimate $p_{\mathrm{succ}}$ respectively.  We describe amplitude estimation more formally:
\begin{Theorem}[Amplitude estimation]~\label{thm:amplitude_estimation}
Suppose $\mathcal{O}_p$ is a unitary with $\mathcal{O}_p\ket{0}_{B}=\ket{\coin{p}}_{B}$. Then there is a unitary procedure $\AEst(\epsilon, \delta)$, making $O(\frac{1}{\epsilon}\log\frac{1}{\delta})$ queries to $\mathcal{O}_p$ and $\mathcal{O}_p^{\dagger}$, that on input $\ket{\coin p}_B\ket{0}_P$ prepares a state of the form
\begin{equation}
    \ket{\coin{p}}_{B}\Bigl(\sum_{p'} \alpha_{p'}\ket{p'}_{P} + \alpha\ket{p_{\perp}}_{P}\Bigr),
\end{equation}
where $\abs{\alpha} \coloneqq \sqrt{1 - \sum_{p'}\abs{\alpha_{p'}}^2} \leq \delta$, $\braket{p'|p_{\perp}} = 0$ for all $p'$, and $\abs{p'-p}\leq \epsilon$ for all $p'$.
\end{Theorem}

Strictly speaking, the parts of \thm{amplitude_estimation} involving $\delta$ come from measuring the output state of the original amplitude estimation procedure~\cite{amplitude_estimation} $O(\log\frac1\delta)$ times and taking the median. This can be made coherent by the principle of deferred measurement.

\subsection{Variable-time amplitude amplification and estimation}\label{sec:VTAA_VTAE}

In this section we review variable-time amplitude amplification (VTAA) and estimation (VTAE), which are essential components of our algorithm.
VTAA and VTAE are procedures applied on top of so-called ``variable-time'' quantum algorithms, which can be formally defined as follows:
\begin{Definition}[Variable-time quantum algorithm, cf.~{{\citealt[Section~3.3]{ambainis2012VTAA}}} and~{{\citealt[Section~5.1]{childs2015quantum}}}]\label{defn:variable_time_algorithm}
Let $\A$ be a quantum algorithm in a space $\H$ that starts in the state $\ket{0}_\H$, the all-zeros state in $\H$. We say $\A$ is a \emph{variable-time quantum algorithm} if the following conditions hold:
\begin{enumerate}
\item $\A$ is the product of $m$ sub-algorithms, $\A = \A_m \A_{m-1} \cdots \A_1$.
\item $\H$ is a tensor product $\H=\H_C\otimes\H_A$, where
$\H_C$ is a tensor product of $m$ single-qubit registers denoted $\H_{C_1}, \H_{C_2}, \ldots, \H_{C_m}$.
\item Each $\A_j$ is a controlled unitary that acts on the registers $\H_{C_j} \otimes \H_A$ controlled on the first $j-1$ qubits of $\H_C$ being set to $\ket{0}$.
\item The final state of the algorithm, $\A\ket{0}_{\H}$, is perpendicular to $\ket{0}_C\coloneqq \ket{0}_{C_1}\ket{0}_{C_2}\cdots\ket{0}_{C_m}$. \label{it:vtqa_final}
\end{enumerate}
\end{Definition}

In each iteration of the variable-time algorithm we shall construct, we use a subroutine that we call \emph{gapped amplitude estimation} ($\GAE$). Standard amplitude estimation~\cite{amplitude_estimation} performs phase estimation on a particular unitary, and $\GAE$ is essentially the same as ``gapped phase estimation''~\cite[Lemma 22]{childs2015quantum} of that unitary. We recall the standard technique of amplitude estimation~\cite{amplitude_estimation}, which we have stated in \thm{amplitude_estimation}. It implies the following:

\begin{Corollary}[Gapped amplitude estimation]\label{cor:gae}
Suppose $\mathcal{O}_p$ is a unitary with $\mathcal{O}_p\ket{0} = \ket{\coin{p}}$. Then there is a unitary procedure $\GAE(\epsilon, \delta; l)$, making $O(\frac{1}{\epsilon}\log\frac{1}{\delta})$ queries to $\mathcal{O}_p$ and $\mathcal{O}_p^{\dagger}$, that on input $\ket{\coin p}_B\ket{0}_{C}\ket{0}_P$, prepares a state of the form
\begin{equation}\label{eq:GAE_state}
\ket{\coin{p}}_B(\beta_0\ket{0}_C\ket{\gamma_0}_P+ \beta_1\ket{1}_C\ket{\gamma_1}_P),
\end{equation}
where $\beta_0,\beta_1\in [0,1]$ satisfy $\beta_0^2 + \beta_1^2=1$ with $\beta_1 \leq \delta$ if $p \geq l - \epsilon$ and $\beta_0 \leq \delta$ if $p < l - 2\epsilon$.
\end{Corollary}
\begin{proof}
We first run $\AEst(\epsilon/4, \delta)$ on registers $B,P$. Then, in register $C$, we output $1$ if the value stored in register $P$ is closer to $l-\epsilon$, and output $0$ if it is closer to $l-2\epsilon$. This gives the desired unitary procedure. For convenience, we put any phase factors on the $\beta_i$ into the $\ket{\gamma_i}$.
\end{proof}

\begin{Theorem}[Variable-time amplitude amplification and estimation~\citealp{ambainis2012VTAA,childs2015quantum, chakraborty2018power}]
\label{thm:vtaa_vtae_append}
Let $\A = \A_m \cdots \A_1$ be a variable-time quantum algorithm on the space $\H=\H_C \otimes\H_F \otimes \H_W$. Let $\ket{0}_\H$ be the all-zeros state in $\H$ and let $t_j$ be the query complexity of the algorithm $\A_j \cdots \A_1$. We define
\begin{equation}
w_j \coloneqq \norm{\Pi_{C_j}\A_j\cdots\A_1\ket{0}_\H}^2 \qquad \mathrm{and} \qquad
t_\mathrm{avg} \coloneqq \sqrt{\textstyle\sum_{j=1}^m w_j t_j^2}
\end{equation}
to be the probability of halting at step $j$ and the root-mean-square average query complexity of the algorithm, respectively,
where $\Pi_{C_j}$ denotes the projector onto $\ket{1}$ in $\H_{C_j}$. We also define
\begin{equation}
\psucc \coloneqq \norm{\Pi_F\A_m\cdots\A_1\ket{0}_\H}^2
\quad \mathrm{and} \quad
\ket{\psi_\mathrm{succ}} \coloneqq \frac{\Pi_F\A_m\cdots\A_1\ket{0}_\H}
{\norm{\Pi_F\A_m\cdots\A_1\ket{0}_\H}}
\end{equation}
to be the success probability of the algorithm and the corresponding output state, respectively, where $\Pi_F$ projects onto $\ket{1}$ in $\H_F$. Then there exists a quantum algorithm that uses $O(Q)$ queries to output the state $\ket{\psi_\mathrm{succ}}$ with probability $\geq1/2$ and a bit indicating whether it succeeds, where
\begin{equation}~\label{eq:vtaa_complexity}
Q \coloneqq t_m \log(t_m) + \frac{t_{\mathrm{avg}}}{\sqrt{\psucc}}\log(t_m).
\end{equation}
There also exists a quantum algorithm that uses $O(\frac{Q}{\epsilon}\log^2 (t_m)\log\log(\frac{t_m}{\delta}))$ queries to estimate $\psucc$ with multiplicative error $\epsilon$ with probability $\geq1-\delta$.
\end{Theorem}

%%%%%%%%%%%%%%%%%%%%%%%%%%%%%%%%%%%%%%%%%%%%%%%%%%%%%%%%%%%%%%%%%%%%%%%%%%%%%%

\subsection{Quantum lower bounds by the adversary method}
Suppose we have $n$ multi-armed bandit oracles $\mathcal{O}_x$,  $x\in[n]$, corresponding to $n$ multi-armed bandits where the best arm is located at a different index in each. Suppose that we also have a best-arm identification algorithm $\A$ that uses no more than $T$ queries to identify the best arm with probability $\geq 1-\delta$.

The basic quantum adversary method~\cite{ambainis2002adversary,spalek2005adversary} considers a quantity of the form
\begin{equation}
    s_k \coloneqq \sum_{x\neq y} w_{x,y} \braket{\psi_x^{(k)}|\psi_y^{(k)}},
\end{equation}
where $k\in \{0,1,\dots,T\}$, $x,y\in [n]$, $w_{x,y}\geq 0$, and $\ket{\psi_x^{(k)}}$ is the state of $\mathcal{A}$ after the $k^\text{th}$ query to the oracle $\mathcal{O}_x$.

At step $k=0$, $\A$ has made no queries to the oracle, so $\ket{\psi_x^{(0)}}$ must be the same for all $x$. Therefore $s_0 = \sum_{x\neq y}w_{x,y}$ as $\braket{\psi_x^{(0)}|\psi_y^{(0)}}=1$.

At step $k=T$, $\A$ must output the index of the best arm with probability $\geq1-\delta$. Since the location of the best arm is different for each $\mathcal{O}_x$, the states $\ket{\psi_x^{(T)}}$ must be distinguishable by a quantum measurement with probability $\geq1-\delta$. This means that $|\braket{\psi_x^{(T)}|\psi_y^{(T)}}|\leq 2\sqrt{\delta(1-\delta)}$. Therefore $\abs{s_T} \leq 2\sqrt{\delta(1-\delta)}\cdot \sum_{x\neq y}w_{x,y}$.

Combining the above observations, we have
\begin{equation}
    \abs{s_0-s_T} \geq \abs{s_0}-\abs{s_T} \geq (1-2\sqrt{\delta(1-\delta)})\cdot \sum_{x\neq y}w_{x,y}.
\end{equation}
Hence, if we can upper bound $\abs{s_{k+1}-s_k}$ by $B$ for some constant $B$, we can deduce that
\begin{equation}
    T \geq \frac{1-2\sqrt{\delta(1-\delta)}}{B} \cdot \sum_{x\neq y}w_{x,y},
\end{equation}
giving a lower bound on the query complexity.

Note that we apply the quantum adversary method to multi-armed bandit oracles of the form given in \eq{quantum-bandit-defn}, whereas most results from the literature on quantum lower bounds assume a different form of oracle. We remark that~\citet{belovs2015variations} treats a more general class of oracles, so it should be possible to prove \thm{quantum-BAI-confidence-lower} using its results. However, we give a self-contained proof using the formulation described above as this approach is straightforward in our case.

%%%%%%%%%%%%%%%%%%%%%%%%%%%%%%%%%%%%%%%%%%%%%%%%%%%%%%%%%%%%%%%%%%%%%%%%%%%%%%

\section{Proof Details of the Quantum Upper Bound}\label{append:upper}

%================================================================

\subsection{Proof of \texorpdfstring{\lem{vtalgo_correctness}}{Lemma \ref{lem:vtalgo_correctness}}}
We first state a more detailed version of \lem{vtalgo_correctness}. We say that states $\ket{\psi}$ and $\ket{\phi}$ are $\epsilon$-close if $\norm{\ket{\psi}-\ket{\phi}} \le \epsilon$.
\begin{Lemma}[Full version of Lemma 1, correctness of $\A$]
The output state $\ket{\phi(\A)}$ of $\A$ is $(\alpha/n)$-close to
\begin{align*}
    \ket{\psi(\A)} &\coloneqq \frac{1}{\sqrt{n}}\sum_{\Sright}\ket{i}_I\ket{\coin
     p_i}_B\ket{\psi_i}_{C,P}\ket{1}_F\\
  &\quad + \frac{1}{\sqrt{n}}\sum_{\Sleft}\ket{i}_I\ket{\coin
    p_i}_B\ket{\psi_i}_{C,P}\ket{0}_F\\
  &\quad + \frac{1}{\sqrt{n}}\sum_{\Smiddle}\ket{i}_I\ket{\coin p_i}_B(\beta_{i,1}\ket{\psi_{i,1}}_{C,P}\ket{1}_F + \beta_{i,0}\ket{\psi_{i,0}}_{C,P}\ket{0}_F)
\end{align*}
for some $\beta_{i,1},\beta_{i,0}\in \mathbb{C}$ and states $\ket{\psi_i}, \ket{\psi_{i,j}}$. In particular, we have $\abs{\psucc-\psucc'}\leq \frac{2\alpha}{n}$ where $\psucc\coloneqq \norm{\Pi_F\ket{\phi(\A)}}^2$ and $\psucc' \coloneqq \norm{\Pi_F\ket{\psi(\A)}}^2 = \frac{1}{n}\bigl(\abs{\Sright}+\sum_{i \in \Smiddle}\abs{\beta_{i,1}}^2\bigr).$
\end{Lemma}

As our proof is similar to that presented in Section 5.3 of~\citet{childs2015quantum}, we only sketch it in a way that highlights the differences. For comparison, it may be helpful to note that our states $\ket{i}_I\ket{\coin p_i}$ are analogous to the matrix eigenstates $\ket{\lambda}$ in~\citet{childs2015quantum}. The controlled-$\NOT$ operation in \lin{vtalgo_forloop4} of our \algo{vtalgo} takes the place of the simulation subroutine called ``$W$'' in Lemma 23 of~\citet{childs2015quantum}, which is much more elaborate.

We proceed with the proof sketch. Let $\A_{\textrm{main}} \coloneqq \A_{m+1}\cdots\A_{1}$ denote the part of $\A$ after initialization. We show that, for each fixed $i$, $\mathcal{A}_{\textrm{main}}\ket{i}_{I}\ket{\coin p_i}_B\ket{0}_{C,P,F}$ is $(\frac{\alpha}{n^{3/2}})$-close to
\begin{enumerate}[wide=0pt, widest={Case $i\in \Smiddle$:}, leftmargin=*]
\item[Case $i\in \Smiddle$:]
$\ket{i}_I\ket{\coin p_i}_B(\beta_{i,1}\ket{\psi_i}_{C,P}\ket{1}_F+\beta_{i,0}\ket{\psi_{i,0}}_{C,P}\ket{0}_F)$ for some $\beta_{i,1},\beta_{i,0}\in \mathbb{C}$ and states $\ket{\psi_i},\ket{\psi_{i,j}}$;
\item[Case $i\in \Sright$:] $\ket{i}_I\ket{\coin p_i}_B\ket{\psi_i}_{C,P}\ket{1}_F$;
\item[Case $i\in \Sleft$:] $\ket{i}_I\ket{\coin p_i}_B\ket{\psi_i}_{C,P}\ket{0}_F$.
\end{enumerate}

Then $\ket{\phi(\A)} = \A\ket{0}_{I,B,C,P,F} = \A_{\text{main}}\frac{1}{\sqrt{n}}\sum_{i=1}^n\ket{i}_I\ket{\coin p_i}_B\ket{0}_{C,P,F}$ is $(\frac{1}{\sqrt{n}}\cdot n\cdot \frac{\alpha}{n^{3/2}} = \frac{\alpha}{n})$-close to $\ket{\psi(\A)}$ as claimed.

\paragraph{Case $i\in \Smiddle$.} This is trivially true because $\beta_{i,1}\ket{\psi_{i,1}}_{C,P}\ket{1}_F+\beta_{i,0}\ket{\psi_{i,0}}_{C,P}\ket{0}_F$ can represent any state on registers $C,P,F$.

\paragraph{Case $i\in \Sleft$.}

Let $j\in [m-1]$ be such that $l_1- 2\epsilon_j \leq p_i < l_1- \epsilon_j$. Note that this $j$ uniquely exists by the definition of $\Sleft$, $m$, and $\epsilon_j$. Then the state of the algorithm after the $(j-1)$st iteration of the for-loop in \lin{vtalgo_forloop} is $(2(j-1)a)$-close to
\begin{equation}
    \ket{i}_I\ket{\coin p_i}_B\ket{0}_C\ket{\gamma_0^1}_{P_1}\cdots\ket{\gamma_0^{j-1}}_{P_{j-1}}\ket{0}_{P_j\cdots P_m}\ket{1}_F,
\end{equation}
where, for each $i$, the state $\ket{0}_{C_i}\ket{\gamma_0}_{P_i}$ corresponds to the state $\ket{0}_C\ket{\gamma_0}$ in $\GAE(\epsilon_j,a;l_1)$. Note that we incur an error of at most $2a$ at each iteration which comes from running $\GAE(\epsilon_j,a;l_1)$ (cf.\ the case where $\beta_1\leq a$ in \cor{gae}). This error accumulates additively.

The state after the $j^\text{th}$ iteration is $(2ja)$-close to
\begin{equation}\label{eq:caught}
\begin{aligned}
    &\beta_0\ket{i}_I\ket{\coin p_i}_B\ket{0}_C\ket{\gamma_0^1}_{P_1}\cdots\ket{\gamma_0^j}_{P_j}\ket{0}_{P_{j+1}\cdots P_m}\ket{1}_F
    \\&+\beta_1 \ket{i}_I\ket{\coin p_i}_B\ket{\boldsymbol{j}}_C\ket{\gamma_0^1}_{P_1}\cdots\ket{\gamma_0^j}_{P_j}\ket{0}_{P_{j+1}\cdots P_m}\ket{1}_F,
\end{aligned}
\end{equation}
where $\boldsymbol{j} \coloneqq 0^{j-1}10^{m-j}$ denotes a unary representation of the integer $j$.

At the $(j+1)$st iteration, the part of the state in the second line of Eq.~\eqref{eq:caught} is unchanged because its register $C$ indicates ``stop'', but the part in the first line of Eq.~\eqref{eq:caught} changes to being $(2(j+1)a)$-close to
\begin{equation}
    \beta_0\ket{i}_I\ket{\coin p_i}_B\ket{\boldsymbol{j+1}}_C\ket{\gamma_0^1}_{P_1}\cdots\ket{\gamma_0^j}_{P_j}\ket{\gamma_0^{j+1}}_{P_{j+1}}\ket{0}_{P_{j+2}\cdots P_m}\ket{0}_F.
\end{equation}
Hence, the state after the $(j+1)$st iteration is $(2(j+1)a)$-close to
\begin{equation}\label{eq:final_state}
\begin{aligned}
    &\beta_0\ket{i}_I\ket{\coin p_i}_B\ket{\boldsymbol{j+1}}_C\ket{\gamma_0^1}_{P_1}\cdots\ket{\gamma_0^j}_{P_j}\ket{\gamma_0^{j+1}}_{P_{j+1}}\ket{0}_{P_{j+2}\cdots P_m}\ket{0}_F\\
     &+\beta_1 \ket{i}_I\ket{\coin p_i}_B\ket{\boldsymbol{j}}_C\ket{\gamma_0^1}_{P_1}\cdots\ket{\gamma_0^j}_{P_j}\ket{0}_{P_{j+1}\cdots P_m}\ket{0}_F.
\end{aligned}
\end{equation}
Since the $C$ register of all parts of the state in Eq.~\eqref{eq:final_state} indicates ``stop'', the remaining iterations $j+2,\ldots, m$ of $\A$ do not alter it. Hence the final state of $\A$ is $(2ma)$-close to the state in Eq.~\eqref{eq:final_state}, which is of the form
\begin{equation}~\label{eq:final_state_simplified}
    \ket{i}_I\ket{\coin p_i}_B\ket{\psi_i}_{C,P}\ket{0}_F.
\end{equation}
Note that $2ma = \frac{\alpha}{n^{3/2}}$, so the closeness of approximation is as claimed.

\paragraph{Case $i\in \Sright$.} In this case, there does not exist a $j\in[m-1]$ such that $l_1-2\epsilon_j\leq p_i< l_1-\epsilon_j$. Thus a simplified version of the argument above, in which we do not have to consider different cases according to the iteration number, shows that the resulting state is $(2ma)$-close to a state of the same form as Eq.~\eqref{eq:final_state_simplified} but with the $F$ register remaining in state $1$.

Lastly, we show that $\psucc$ is close to $\psucc'$ as claimed:
\begin{equation}
\begin{aligned}
    \abs{\psucc-\psucc'} &= \abs{\bigl(\sqrt{\psucc}+\sqrt{\psucc'}\bigr)\cdot\bigl(\sqrt{\psucc}-\sqrt{\psucc'}\bigr)}\\[0.4em]
    &= \bigl(\sqrt{\psucc}+\sqrt{\psucc'}\bigr)\cdot \,\bigl|{\norm{\Pi_F\ket{\phi(\A)}}-\norm{\Pi_F\ket{\psi(\A)}}}\bigr|\\[0.4em]
    &\leq 2\,{\norm{\Pi_F(\ket{\phi(\A)}-\ket{\psi(\A)})}}\\[0.4em]
    &\leq 2\, \frac{\alpha}{n}.
\end{aligned}
\end{equation}

\subsection{Proof of \texorpdfstring{\lem{vtalgo_complexity}}{Lemma \ref{lem:vtalgo_complexity}}}
The proof is similar to that presented in Section 5.4 of~\citet{childs2015quantum}. For the first claim, note first that $\A_0$ and $\A_{m+1}$ use a constant number of queries ($1$ and $0$, respectively), so we can ignore them. For $k\in [m]$, $\A_k$ only uses queries to perform $\GAE(\epsilon_k,d;l_1)$, which takes $O(\frac{1}{\epsilon_k}\log\frac{1}{a})$ queries. Therefore $t_j$, the number of queries in $\A_j \A_{j-1}\cdots \A_1$, is of order
\begin{equation}
    \sum_{k=1}^j\frac{1}{\epsilon_k}\log\Bigl(\frac{1}{a}\Bigr) =  \sum_{k=1}^j2^k\log\Bigl(\frac{1}{a}\Bigr) \leq 2^{j+1}\log\Bigl(\frac{1}{a}\Bigr)
\end{equation}
because $\epsilon_k = 2^{-k}$ by definition. In addition, we have $t_m = O(\frac{1}{\Delta}\log\frac{1}{a})$ because $m = \ceil{\log\frac{1}{\Delta}}+2$ by definition. The first claim follows.

For the second claim, we have
\begin{align}\label{eq:t_avg_computation}
    t_{\mathrm{avg}}^2 = \sum_{j=1}^mw_jt_j^2 &=  \sum_{j=1}^m \biggl\|\Pi_{C_j}\A_j\cdots\A_1\frac{1}{\sqrt{n}}\sum_{i=1}^n\ket{i}_I\ket{\coin p_i}_B\ket{0}_C\ket{0}_P\ket{1}_F\biggr\|^2t_j^2 \\
    &= \frac{1}{n} \sum_{i=1}^n \sum_{j=1}^m w_{i,j}t_j^2\\
    &= \frac{1}{n}\sum_{i=1}^n\tau^2_i,
    \label{eq:tavg_decomp}
\end{align}
where $w_{i,j}\coloneqq\norm{\Pi_{C_j}\A_j\cdots\A_1\ket{i}_I\ket{\coin p_i}_B\ket{0}_C\ket{0}_P\ket{1}_F}^2\in[0,1]$ and $\tau_i \coloneqq \sum_{j=1}^m w_{i,j}t_j^2$.

Note that $w_{i,j}$ can be thought of as the probability that $\mathcal{A}$ stops at the end of iteration $j$ if initialized with arm $i$; $\tau_i^2$ can be thought of as the squared average stopping time of $\mathcal{A}$ if initialized with arm $i$.

For each fixed $i$, we consider $\tau_i^2$ according to the following three cases.

\paragraph{Case $i\in\Sright$.}
We have $\sum_{j=1}^m w_{i,j}=1$, so $\tau_i^2 \leq t_m^2 = O(2^{2m}\log^2(\frac{1}{a}))=O(\frac{1}{\Delta^2}\log^2(\frac{1}{a}))$ because $m=\ceil{\log\frac{1}{\Delta}}+2$ by definition.

\paragraph{Case $i\in \Smiddle$.} We still have $\tau_i^2 = O(\frac{1}{\Delta^2}\log^2(\frac{1}{a}))$ as in the case $i\in\Sright$, by exactly the same argument. But by the definition of $\Smiddle$, we have $l_1-p_i\leq\Delta/2$, so we can also write $\tau_i^2 = O\bigl(\frac{1}{(l_1-p_i)^2}\log^2(\frac{1}{a})\bigr)$.

\paragraph{Case $i\in \Sleft$.}
For $i\in \Sleft$, let $j\in [m-1]$ be such that $l_1- 2\epsilon_j \leq p_i < l_1- \epsilon_j$ as in the proof of \lem{vtalgo_correctness}.

We know that after the $(j+1)$st iteration, the state is $(ma = \alpha/n)$-close to the state in \eq{final_state} on which the algorithm terminates. Therefore, the probability $w_{i,j+1}$ of terminating after the $(j+1)^\text{st}$ iteration is $1-O((\alpha/n)^2)$. It can also be seen that the probability $w_{i,j+r}$ of terminating after the $(j+r)^\text{th}$ iteration is $(1- O((\alpha/n)^2))\cdot O((\alpha/n)^{2(r-1)})$. Hence
\begin{equation}
    \tau_i^2 \leq t_{j+1}^2 + O\Bigl(\sum_{r=2}^{m-j} \Bigl(\frac{\alpha}{n}\Bigr)^{2(r-1)} t_{j+r}^2\Bigr) = O(t_{j+1}^2)= O\biggl(\frac{\log^2(\frac{1}{a})}{\epsilon_{j+1}}\biggr) =  O\biggl(\frac{\log^2(\frac{1}{a})}{(l_1-p_i)^2}\biggr),
\end{equation}
where we used $\epsilon_{j+1}=\epsilon_{j}/2 \geq (l_1-p_i)/4$ for the last inequality.

Substituting the above results into \eq{tavg_decomp} tells us that $t_{\mathrm{avg}}^2$ is of order
\begin{equation}~\label{eq:t_avg_appendix}
\frac{1}{n}\Bigl(\frac{\abs{\Sright}}{\Delta^2} + \sum_{i\in \Sleft\cup \Smiddle} \frac{1}{(l_1-p_i)^2}\Big)\cdot \log^2\Bigl(\frac{1}{a}\Bigr)
\end{equation}
as desired.

\subsection{Proof of \texorpdfstring{\lem{vtaa_vtae_on_vtalgo}}{Lemma \ref{lem:vtaa_vtae_on_vtalgo}}}
We set the approximation parameter in $\A$ to be $\alpha = c\delta$ for some constant $c < 0.05$ to be determined later. Then $\alpha<0.05$.

We apply VTAA (\thm{vtaa_vtae_append}) on $\A$. This gives an algorithm that outputs a state $\ket{\psi_{\mathrm{succ}}}$ that is $(\frac{\alpha}{n} = \frac{c\delta}{n})$-close to the (normalized) state proportional to
\begin{align}
\Pi_F\ket{\psi(\A)} = \frac{1}{\sqrt{n}}\Bigl(\sum_{i \in \Sright}\!\!\ket{i}_I\ket{\coin
 p_i}_B\ket{\psi_i}_{C,P}\ket{1}_F
+\!\!\sum_{i \in \Smiddle}\!\!\alpha_{i,1}\ket{i}_I\ket{\coin p_i}_B\ket{\psi_{i,1}}_{C,P}\ket{1}_F\Bigr)
\end{align}
with success probability at least $1/2$ and a bit indicating success or failure. Now, we repeat the entire procedure $O(\log\frac{1}{\delta})$ times to prepare $\ket{\psi_{\text{succ}}}$ at least once with probability $\geq 1-\delta/2$. Once $\ket{\psi_{\text{succ}}}$ has been successfully prepared, as indicated by the algorithm, we measure its index register $I$. This procedure outputs an arm index in $\Sright \cup \Smiddle$ with probability $\geq(1-\delta/2)\cdot (1-2c\delta/n)$ which is $\geq1-\delta$ for $c\leq 1/4$ sufficiently small. So, as we also need $c<0.05$, we choose $c=0.01$. We call this procedure $\Amplify(\A, \delta)$.

Let us consider the query complexity of $\Amplify(\A, \delta)$. We have
\begin{equation}~\label{eq:vtaa_vtae_tm}
    t_{m+1}  = O\Bigl(\frac{1}{\Delta}\log\Bigl(\frac{1}{a}\Bigr)\Bigr) = O\Bigl(\frac{1}{\Delta}\log\Bigl(n\log\Bigl(\frac{1}{\Delta}\Bigr)\Bigr)\Bigr)
\end{equation}
because $a = \frac{\alpha}{2 (\ceil{\log(1/\Delta)}+2)n^{3/2}}$ by definition. We also have
\begin{equation}
    \psucc \geq  \psucc'-\frac{2\alpha}{n}\geq \frac{\abs{\Sright}}{n}-\frac{0.1}{n}> \frac{\abs{\Sright}}{2n},
\end{equation}
where we used the assumption $\abs{\Sright}>0$ for the last inequality. Lastly, $t_{\mathrm{avg}}^2$ is of order given in \eq{t_avg} (reproduced in \eq{t_avg_appendix} above). Therefore, substituting all these bounds into \eq{vtaa_complexity} of \thm{vtaa_vtae_append}, we see that $\Amplify(\A, \delta)$ has query complexity of order
\begin{equation}\label{eq:detailed_vtaa_query_general}
\biggl(\frac{1}{\Delta^2} + \frac{1}{\abs{\Sright}}\sum_{\Sleft\cup \Smiddle} \frac{1}{(l_1-p_i)^2}\biggr)\cdot \log\Bigl({\frac{n}{\delta}\log{\frac{1}{\Delta}}}\Bigr)\cdot \log\Bigl(\frac{1}{\Delta}
    \log\Bigl(\frac{n}{\delta}\log\Bigl(\frac{1}{\Delta}\Bigr)\Bigr)\Bigr)\cdot \log\Bigl(\frac{1}{\delta}\Bigr).
\end{equation}

We also apply VTAE (\thm{vtaa_vtae_append}) with multiplicative accuracy $\epsilon$ and confidence $\delta$ on $\A$. This gives an algorithm, $\Estimate(\A,\epsilon,\delta)$, that outputs an estimate $r$ of $\psucc$ with multiplicative accuracy $\epsilon$ (i.e., $\abs{r-\psucc}<\epsilon \psucc$) with probability $\geq 1-\delta$. Combining $\abs{r-\psucc}<\epsilon \psucc$ with $\abs{\psucc-\psucc'}\leq \frac{2\alpha}{n} < \frac{0.1}{n}$ gives
\begin{equation}
(1-\epsilon)\Bigl(\psucc'-\frac{0.1}{n}\Bigr) < r < (1+\epsilon)\Bigl(\psucc' + \frac{0.1}{n}\Bigr)
\end{equation}
as claimed.

The query complexity of $\Estimate(\A,\epsilon,\delta)$ is given by \eq{detailed_vtaa_query_general} times
\begin{equation}\label{eq:detailed_vtae_query_general}
   \frac{1}{\epsilon} \log^2(t_{m+1})\log\Bigl(\log\Bigl(\frac{t_{m+1}}{\delta}\Bigr)\Bigr) =   O\Bigl(\frac{1}{\epsilon}\, \poly\Bigl(\log\Bigl(\frac{n}{\delta \Delta}\Bigr)\Bigr)\Bigr)
\end{equation}
according to \thm{vtaa_vtae_append} and \eq{vtaa_vtae_tm}.

\subsection{Proof of \texorpdfstring{\lem{locate}}{Lemma \ref{lem:locate}}}

From the first claim of \lem{shrink}, we see that the probability of $E^c$ is at most
$\frac{\delta}{4}\sum_{i=0}^{\infty}2^{-i} = \delta/2$, where the geometric series arises because of \lin{locate_delta_halve}. Henceforth, we assume $E$.

Consider the first claim. For given intervals $I_2$, $I_1$, let us write
\begin{equation}
    \gap(I_2,I_1) \coloneqq \min I_1 - \max I_2.
\end{equation}
At the end of iteration $i\geq 1$ (i.e., after \lin{locate_delta_halve}), we have $|I_k| = (3/5)^i$ by the first claim of \lem{shrink}. At the end of iteration $\ceil{\log_{5/3}(\frac{1}{\Delta_2})}+3$, we have $\abs{I_k}< \Delta_2/4$, so $\gap(I_2,I_1) > \Delta_2 - 2\Delta_2/4 = \Delta_2/2 > 2\abs{I_1}$ because $p_k\in I_k$. Therefore the while loop must break at this point if it has not done so earlier. For the returned $I_k$, we clearly have $p_k\in I_k$ because $E$ holds, and $ \gap(I_2,I_1) > 2\abs{I_1}$ because the while loop has broken. During the while loop, because $\abs{I_k}$ decreases from iteration to iteration, we always have $\abs{I_k} \geq  (3/5)^{\ceil{\log_{5/3}(\Delta_2^{-1})}+3}\geq \Delta_2/8$. Note that $\abs{I_1}=\abs{I_2}$ because, at each iteration of the while loop, the $\Shrink$ subroutine always shrinks intervals by the same factor of $3/5$ and $\abs{I_1}=\abs{I_2}=1$ initially.

Now, consider the second claim. From the first claim, we know that the while loop breaks at or before the end of iteration $\ceil{\log_{5/3}(\Delta_2^{-1}))}+3$, and we always have
$1/\delta_i = O(2^{\log_{5/3}({\Delta}_2^{-1})}/\delta) = O(\Delta_2^{-2}/\delta)$, where $\delta_i = \delta/2^{2+i}$ is the confidence parameter in $\Shrink$ at iteration $i$. Therefore, using the second claim of \lem{shrink}, the total number of queries used is at most
\begin{equation}\label{eq:bestarm_complexity}
    O(\log(\Delta_2^{-1})) \cdot O\Bigl(\sqrt{H}\cdot \poly\Bigl(\log\Bigl(\frac{n}{\Delta_2}\cdot \frac{\Delta_2^{-2}}{\delta}\Bigr)\Bigr)\Bigr),
\end{equation}
which is $O\bigl(\sqrt{H}\cdot \poly(\log(\frac{n}{\delta \,\Delta_2}))\bigr)$ as desired.

\subsection{Proof of \texorpdfstring{\lem{shrink}}{Lemma\ref{lem:shrink}}}

Throughout, we fix $k\in \{0,1\}$.

For the first claim, it is clear that $\abs{J}=3\abs{I}/5$ because all the intervals appearing in Lines~\ref{lin:shrink_switch_case1}--\ref{lin:shrink_switch_case4} have length $3\epsilon$. Our proof that $p_k\in J$ with high probability is similar to that in Section 4 of~\citet{lintong} so we only present a brief sketch below.

Let us write $x_j = a + j \epsilon$ for $j=0,\ldots,5$, so that $x_0 = a$ and $x_5 = b$. Let $E$ be the event that both $\Estimate$s in Lines~\ref{lin:shrink_estimate1} and \ref{lin:shrink_estimate2} return the correct result. The probability of $E^{c}$ is at most $\delta$ so we restrict to the case of $E$ in the following paragraph.

For $j\in \{1,2\}$, we can use \eq{vtae_estimate_quality} in \lem{vtaa_vtae_on_vtalgo} to see that if $p_k \leq x_j$, then $B_j = 0$ because $r_j\leq (1+0.1)(\frac{k}{n+1}+\frac{0.1}{n+1}) < \frac{k+0.5}{n+1}$, whereas if $p_k \geq x_{j+2}$, then $B_j = 1$ because $r_j\geq (1-0.1)(\frac{k+1}{n+1}-\frac{0.1}{n+1}) > \frac{k+0.5}{n+1}$. Here we use the fact $k\in \{1,2\}$. By considering the contrapositive of the previous two if-then statements, we establish the first claim.

For more details, we refer the reader to Section 4 of~\citet{lintong}, in particular its Table~2 and Algorithm~1. Note that in the case of $(B_1,B_2) = (0,1)$, we could have shrunk the interval to $[a+2\epsilon, a+3\epsilon]$ and still maintained $p_k\in J$, as is done in~\citet{lintong}. However, it is important for us to keep the shrinkage factor ($3/5$) the same in all cases because we use this to prove correctness in \lem{locate}.

We now prove the second claim. Since we run $\Estimate$ with constant multiplicative error $\epsilon = 0.1$, its query complexity is of order \eq{detailed_vtaa_query_general}, which is
\begin{equation}
    \frac{1}{\Delta^2} + \frac{1}{\abs{\Sright}}\sum_{i \in \Sleft\cup \Smiddle} \frac{1}{(l_1-p_i)^2}
\end{equation}
up to polylog factors, where we recall that $\Delta = l_1-l_2$. In addition, we recall
\begin{equation}
    \Sleft \cup \Smiddle = \{i: p_i < l_1 - \Delta/8\}
\end{equation}
from \eq{Smiddle} and \eq{Sright}. Note that $\abs{\Sright}>0$ because we appended an arm with bias $p_0=1$.

By assumption, $\abs{I}\geq \Delta_2/8$. So, in view of Lines~\ref{lin:shrink_vtalgo1} and \ref{lin:shrink_vtalgo2}, we have $\Delta = 2\epsilon = 2\abs{I}/5 \geq \Delta_2/20$. Therefore $1/\Delta^2 = O(1/\Delta_2^2)$.

We also need to compare $p_1-p_i$ with $l_1-p_i$ for $i\in \Sleft \cup \Smiddle$. By definition, we have $p_i < l_1 - \Delta/8$, so $l_1-p_i > \Delta/8$. Note that we also have $\abs{p_k-l_1} \leq \abs{I} = 5\Delta/2$ because $p_k\in I$ by assumption and $l_1\in I$ by definition. If $k=1$, this says $\abs{p_1-l_1} \leq 5\Delta/2$. If $k=2$, this says $\abs{p_2-l_1} \leq 5\Delta/2$, but we can still bound
\begin{equation}
    \abs{p_1-l_1}\leq \Delta_2+\abs{p_2-l_1} \leq 20\Delta+5\Delta/2 < 25\Delta.
\end{equation}
So regardless of whether $k=1$ or $k=2$, we have that $\abs{p_1-l_1}<25\Delta$. Therefore
\begin{equation}
    \frac{p_1-p_i}{l_1-p_i} = 1 + \frac{p_1-l_1}{l_1-p_i} < 1 + \frac{25\Delta}{\Delta/8} = 201,
\end{equation}
and so $1/(l_1-p_i)^2 = O(1/(p_1-p_i)^2)$. Hence we have established the second claim.

%%%%%%%%%%%%%%%%%%%%%%%%%%%%%%%%%%%%%%%%%%%%%%%%%%%%%%

\section{Corollaries for the  Fixed-budget Setting}\label{append:pac_fixed_budget}
As mentioned near the end of the main body, by using a reduction similar to that from Monte Carlo to Las Vegas algorithms, we can construct a fixed-budget algorithm from our fixed-confidence one. For completeness, we state and prove the following result:

\begin{Lemma}[Reduction to fixed confidence]\label{lem:fixed_confidence_reduction}
Let $\mathcal{O}$ be a multi-armed bandit oracle. Suppose that for any $\delta\in(0,1)$, we have an algorithm  $\mathcal{A}_c(\delta)$ that with probability $\geq1-\delta$, terminates before using $T_c(\delta)$ queries to $\mathcal{O}$ and returns the best-arm index $i^*=1$. Suppose that we also know $T_c(\delta)$.
Then, for any positive integer $T$, we can construct an algorithm $\mathcal{A}_b(T)$ that returns $i^*=1$ with probability $\geq \min_{\delta\in(0,1)}\exp\bigl(-\lfloor{T/T_c( \delta)}\rfloor D(\frac{1}{2}\|\delta)\bigr)$ using at most $T$ queries to $\mathcal{O}$, where $D(p\|q)$ is the relative entropy between Bernoulli random variables with bias $p$ and $q$.
\end{Lemma}

\begin{proof}
Since $T_c( \delta)$ is known, consider the modified version of the
fixed-confidence algorithm where the algorithm is forced to halt and
return some blank symbol ``$\bot$'' if the running time exceeds $T_c(\delta)$.
We refer to the modified algorithm as  $A^\prime_c(\delta)$. $A'_c(\delta)$ returns the best-arm index $i^*=1$ with probability $\geq1 - \delta$ and returns some symbol in $\{2, \ldots, n, \bot\}$ with probability $\leq\delta$.

For any $T$, we construct $\mathcal{A}_b(T)$ as follows. Pick some $\delta\in (0,1)$, run $A'_c(\delta)$ $m \coloneqq \lfloor{T/T_c(\delta)}\rfloor$ times, and take a majority vote over the outcomes. The failure probability can be upper bounded by the probability that $i^*$ is observed fewer than $m/2$ times. The Chernoff bound upper bounds the latter probability by $\exp(-m D(\frac{1}{2}\|\delta)) = \exp\bigl(-\lfloor{T/T_c( \delta)}\rfloor D(\frac{1}{2}\|\delta)\bigr)$. But $\delta$ was arbitrary, so we can take the $\delta$ that minimizes this upper bound.
\end{proof}

As a direct corollary of \thm{quantum-BAI-confidence} and \lem{fixed_confidence_reduction}, we see that when $H$ (therefore $T_c$) is known in advance, for sufficiently large $T$, there is a quantum algorithm using at most $T$ queries that returns the best arm with probability  $\geq 1 - \exp(-\Omega(T/\sqrt{H}))$.

%%%%%%%%%%%%%%%%%%%%%%%%%%%%%%%%%%%%%%%%%%%%%%%%%%%%%%%%%%%%%%%%%%%%%%%%%%%%%%

\section{Proof Details of the Quantum Lower Bound}\label{append:lower}

\subsection{Proof of \texorpdfstring{\thm{quantum-BAI-confidence-lower}}{Theorem \ref{thm:quantum-BAI-confidence-lower}}}

For convenience, we reproduce the statement of the result:
\quantumlower*

\begin{proof}
We use the adversary method and consider the following $n$ different multi-armed bandit oracles.

In the $1^\text{st}$ bandit, we assign bias $p_i$ to arm $i$. Let $\eta>0$ be a constant to be determined later. In the $x^\text{th}$ bandit, $x\in \{2,\ldots,n\}$, we assign bias $p_1' \coloneqq p_1+\eta$ to arm $x$ and $p_i$ to arm $i$ for all $i\neq x$. A best-arm identification algorithm must output arm $x$ on assignment $x$ for all $x\in [n]$ with probability $\geq1-\delta$.

Following the adversary method, we consider the sum
\begin{equation}
s_k \coloneqq \sum_{x > 1} \frac{1}{\Delta_x'^2}\braket{\psi_x^{(k)}|\psi_1^{(k)}}
\end{equation}
for $x \in [n]$, where $\Delta_x'\coloneqq p_1'-p_x$. Clearly
\begin{equation}\label{eq:start_delta_ps}
s_0 = \sum_{x>1}\frac{1}{\Delta_x'^2}.
\end{equation}
We also have
\begin{equation}\label{eq:end_delta_ps}
s_T \leq \sum_{x>1}\frac{1}{\Delta_x'^2}\cdot 2\sqrt{\delta(1-\delta)}.
\end{equation}

Next, we bound the difference $\abs{s_{k+1}-s_{k}}$. For $i>1$, we let
\begin{equation}
\renewcommand\arraystretch{1.5}
A_i \coloneqq \begin{pmatrix}
\sqrt{1-p_i} & \sqrt{\vphantom{1-}p_i}\\
\sqrt{\vphantom{1-}p_i} & - \sqrt{1-p_i}
\end{pmatrix},
\end{equation}
while
\begin{equation}
\renewcommand\arraystretch{1.5}
A_1 \coloneqq \begin{pmatrix}
\sqrt{1-p_1'} & \sqrt{\vphantom{1-}p_1'}\\
\sqrt{\vphantom{1-}p_1'} & - \sqrt{1-p_1'}
\end{pmatrix},
\end{equation}
where we recall $p_1' = p_1+\eta$ by definition.

Now, let us write
\begin{equation}
\ket{\psi_x^{(k)}} = \sum_{z,i,b} \alpha_{x,z,i,b}\ket{z,i,b},\quad         \ket{\psi_1^{(k)}} = \sum_{z,i,b} \alpha_{1,z,i,b}\ket{z,i,b}.
\end{equation}
Then
\begin{equation}\label{eq:x_step}
\ket{\psi_x^{(k+1)}} = \mathcal{O}_x\ket{\psi_x^{(k)}}
= \sum_{z,b} \alpha_{x,z,x,b}\ket{z,x}A_1\ket{b} + \sum_{i\neq x}\sum_{z,b}\alpha_{x,z,i,b}\ket{z,i}A_i\ket{b}
\end{equation}
and similarly
\begin{equation}\label{eq:1_step}
\ket{\psi_1^{(k+1)}} = \mathcal{O}_1\ket{\psi_1^{(k)}}= \sum_{z,b} \alpha_{1,z,x,b}\ket{z,x}A_x\ket{b} + \sum_{i\neq x}\sum_{z,b}\alpha_{1,z,i,b}\ket{z,i}A_i\ket{b}.
\end{equation}
Then
\begin{equation}\label{eq:progress_step_initial_ps}
    \abs{s_{k+1} - s_k} \leq  \sum_{x>1} \frac{1}{\Delta_x'^2}\abs{ \bra{\psi_x^{(k)}}\mathcal{O}_x^{\dagger}\mathcal{O}_1\ket{\psi_1^{(k)}}
    -
    \braket{\psi_x^{(k)}|\psi_1^{(k)}}}.
\end{equation}

Using \eq{x_step} and \eq{1_step}, and after cancellations, we find that
\begin{equation}\label{eq:progress_step_ps}
\begin{aligned}
    \bra{\psi_x^{(k)}}\mathcal{O}_x^{\dagger}\mathcal{O}_1\ket{\psi_1^{(k)}} -  \braket{\psi_x^{(k)}|\psi_1^{(k)}}&=
    \sum_{z,b,b'} \alpha_{x,z,x,b}^*\alpha_{1,z,x,b'}\bra{b}(A_1^{\dagger}A_x-\mathbb{I})\ket{b'}.
\end{aligned}
\end{equation}
With
\begin{equation}
\begin{aligned}
\begin{pmatrix}
u_x & v_x\\
-v_x & u_x
\end{pmatrix} &\coloneqq
A_1^{\dagger}A_x - \mathbb{I} \\
&= \renewcommand\arraystretch{1.75}
\begin{pmatrix}
\sqrt{(1-p_1')(1-p_x)}+\sqrt{p_1'p_x}-1 & \sqrt{(1-p_1')p_x}-\sqrt{p_1'(1-p_x)}\\
-\sqrt{(1-p_1')p_x} + \sqrt{p_1'(1-p_x)} & \sqrt{(1-p_1')(1-p_x)}+\sqrt{p_1'p_x}-1
\end{pmatrix},
\end{aligned}
\end{equation}
\renewcommand\arraystretch{1}
we have
\begin{equation}\label{eq:progress_step_pre_cauchy_ps}
\begin{aligned}
    \abs{s_{k+1} - s_k} &\leq \sum_{x>1}\sum_{z,b}\frac{\abs{u_x}}{\Delta_x'^2}\abs{\alpha_{x,z,x,b}}\abs{\alpha_{1,z,x,b}} + \sum_{x>1}\sum_{z,b\neq b'}\frac{\abs{v_x}}{\Delta_x'^2}\abs{\alpha_{x,z,x,b}}\abs{\alpha_{1,z,x,b'}}.
\end{aligned}
\end{equation}

Clearly, $|u_x| = 1- \sqrt{(1-p_1')(1-p_x)} - \sqrt{p_1'p_x}\leq 1-(1-p_1') - p_x = p_1'-p_x = \Delta_x'$. It can also be seen that $|v_x| \leq \Delta_x'/c(p-\eta)$, where $c(x) \coloneqq 2\sqrt{x(1-x)}$ is a monotone increasing function when $x\in [0,1/2]$. For completeness, we prove the latter inequality as an auxiliary \lem{inequality} immediately after this proof.

We can establish the following bounds using the Cauchy-Schwarz inequality:
\begin{equation}
\begin{aligned}
    \sum_{x>1}  \sum_{z,b}\frac{\abs{u_x}}{\Delta_x'^2}\abs{\alpha_{x,z,x,b}}\abs{\alpha_{1,z,x,b}}
    &\leq \sqrt{\sum_{x>1,z,b} \frac{\abs{u_x}^2}{\Delta_x'^4}\abs{\alpha_{x,z,x,b}}^2} \cdot
    \sqrt{\sum_{x>1,z,b} \abs{\alpha_{1,z,x,b}}^2}\\
    &\leq \sqrt{\sum_{x>1} \frac{1}{\Delta_x'^2}}
\end{aligned}
\end{equation}
and
\begin{equation}
\begin{aligned}
    \sum_{x>1} \sum_{z,b\neq b'}\frac{\abs{v_x}}{\Delta_x'^2}\abs{\alpha_{x,z,x,b}}\abs{\alpha_{1,z,x,b'}} &= \sum_{b\neq b'}\sum_{x>1, z} \frac{\abs{v_x}}{\Delta_x'^2}\abs{\alpha_{x,z,x,b}}\abs{\alpha_{1,z,x,b'}}\\
    &\leq \sum_{b\neq b'}\sqrt{\sum_{x>1, z} \frac{\abs{v_x}^2}{\Delta_x'^4}\abs{\alpha_{x,z,x,b}}^2}\cdot
    \sqrt{\sum_{x>1, z} \abs{\alpha_{1,z,x,b'}}^2}\\
    &\leq \frac{2}{c(p-\eta)} \sqrt{\sum_{x>1} \frac{1}{\Delta_x'^2}}.
\end{aligned}
\end{equation}
Therefore, we find that
\begin{equation}\label{eq:progress_delta_ps}
\abs{s_{k+1}-s_k} \leq \biggl(1+ \frac{2}{c(p-\eta)}\biggr)\sqrt{\sum_{x>1} \frac{1}{\Delta_x'^2}}.
\end{equation}
Hence, from Eqs.~\eqref{eq:start_delta_ps}, \eqref{eq:end_delta_ps}, and \eqref{eq:progress_delta_ps}, we find that
\begin{equation}
T \geq \frac{1-2\sqrt{\delta(1-\delta)}}{1+2/c(p-\eta)} \sqrt{\sum_{x>1}\frac{1}{\Delta_x'^2}}.
\end{equation}

We then set $\eta = p(p_1-p_2)/2$. Now, it can be seen that
\begin{equation}
c(p-\eta) = c\Bigl(\Bigl(1-\frac{p_{1}-p_{2}}{2}\Bigr)p\Bigr)\geq c(p/2)
\end{equation}
because $p\leq 1/2$ and $p_{1}-p_{2}\leq 1$. Moreover,  for $x>1$,
\begin{equation}
\Delta_x' = p_1+\eta-p_x = \frac{p}{2}(p_{1}-p_{2})+(p_1-p_x) \leq \Bigl(1+\frac{p}{2}\Bigr)(p_1-p_x) \leq \frac{5}{4}\Delta_x
\end{equation}
because $p_{x}\leq p_{2}$ and $p\leq 1/2$. Therefore, we find that
\begin{equation}
T \geq \frac{4}{5}\cdot \frac{1-2\sqrt{\delta(1-\delta)}}{1+2/c(p/2)} \sqrt{\sum_{x>1}\frac{1}{\Delta_x^2}},
\end{equation}
and hence $T = \Omega\bigl(\sqrt{\sum_{i=2}^n\frac{1}{\Delta_i^2}}\bigr)$.
\end{proof}

\begin{Lemma}\label{lem:inequality}
Suppose that $p_{1},p_{2}\in [p,1-p]$ where $0<p\leq1/2$. Then
\begin{align}
|\sqrt{(1-p_{1})p_{2}}-\sqrt{(1-p_{2})p_{1}}|\leq\frac{|p_{1}-p_{2}|}{2\sqrt{p(1-p)}},
\end{align}
and the term in the denominator is optimal.
\end{Lemma}

\begin{proof}
Note that
\begin{align}
\sqrt{(1-p_{1})p_{2}}-\sqrt{(1-p_{2})p_{1}}&=\frac{(1-p_{1})p_{2}-(1-p_{2})p_{1}}{\sqrt{(1-p_{1})p_{2}}+\sqrt{(1-p_{2})p_{1}}}\\
&=\frac{-(p_{1}-p_{2})}{\sqrt{(1-p_{1})p_{2}}+\sqrt{(1-p_{2})p_{1}}}.
\end{align}
Therefore, it suffices to prove
\begin{align}\label{eq:inequality-1}
\sqrt{(1-p_{1})p_{2}}+\sqrt{(1-p_{2})p_{1}}\geq 2\sqrt{p(1-p)}.
\end{align}
Since $p_{1},p_{2}\in [p,1-p]$, we have
\begin{align}
(p_{1}-p)(p_{1}-(1-p)) &\leq 0 \label{eq:bnd1} \\
(p_{2}-p)(p_{2}-(1-p)) &\leq 0 \label{eq:bnd2} \\
|2p_{1}-1| &\leq 1-2p \label{eq:bnd3} \\
|2p_{2}-1| &\leq 1-2p. \label{eq:bnd4}
\end{align}
Eqs.~\eqref{eq:bnd1} and \eqref{eq:bnd2} are equivalent to
\begin{align}\label{eq:inequality-2}
p_{1}-p_{1}^{2}\geq p(1-p),\quad p_{2}-p_{2}^{2}\geq p(1-p).
\end{align}
Eqs.~\eqref{eq:bnd3} and \eqref{eq:bnd4} imply
\begin{align}
4p_{1}p_{2}-2p_{1}-2p_{2}+1=(2p_{1}-1)(2p_{2}-1)\leq (2p-1)^{2}=4p^{2}-4p+1,
\end{align}
which gives
\begin{align}\label{eq:inequality-3}
p_{1}+p_{2}-2p_{1}p_{2}\geq 2p-2p^{2}.
\end{align}

Now, we have
\begin{align}
\Bigl(\sqrt{(1-p_{1})p_{2}}+\sqrt{(1-p_{2})p_{1}}\Bigr)^{2}&=(1-p_{1})p_{2}+(1-p_{2})p_{1}+2\sqrt{(1-p_{1})p_{2}(1-p_{2})p_{1}} \\
&=p_{1}+p_{2}-2p_{1}p_{2}+2\sqrt{p_{1}(1-p_{1})}\sqrt{p_{2}(1-p_{2})} \\
&\geq 2p-2p^{2}+2p(1-p)=(2\sqrt{p(1-p)})^{2},
\end{align}
where the inequality comes from \eq{inequality-2} and \eq{inequality-3}. Therefore, we have established \eq{inequality-1}. Note that this is optimal as taking $p_{1}=p_{2}=p$ makes the two sides in \eq{inequality-1} equal.
\end{proof}

\end{document}